\documentclass{article}
\usepackage{graphicx} 

\usepackage{amsmath,amssymb,amsthm}
\usepackage{mathtools}
\theoremstyle{plain}
\newtheorem{theorem}{Theorem}[section]
\newtheorem{proposition}[theorem]{Proposition}
\newtheorem{lemma}[theorem]{Lemma}

\theoremstyle{definition}
\newtheorem{definition}[theorem]{Definition}

\theoremstyle{remark}

\usepackage{microtype} 

\usepackage{bm} 
\usepackage{booktabs} 
\usepackage{multirow}
\usepackage{tabularx} 
\newcolumntype{C}{>{\centering\arraybackslash}X} 
\usepackage{enumitem} 
\usepackage{hyperref} 
\usepackage{xcolor} 

\usepackage{ifthen}
\newboolean{mybool}
\setboolean{mybool}{true} 

\usepackage{todonotes}

\usepackage{booktabs}  
\usepackage{colortbl}  

\usepackage{subcaption}

\usepackage{etoolbox} 
\newtoggle{anonymous}
\newtoggle{comments}
\newtoggle{short}
\newtoggle{abstract}
\newtoggle{neurips}
\togglefalse{anonymous}
\togglefalse{comments}
\togglefalse{short}
\togglefalse{abstract}
\togglefalse{neurips}

\iftoggle{neurips}{
\usepackage[preprint]{neurips_2024}
}{
\usepackage[accepted]{icml2025}
}

\iftoggle{comments}{
  \newcommand{\alex}[2][]{\todo[color=green!50,#1]{\textsf{Alex:} #2}}
  \newcommand{\anti}[2][]{\todo[color=red!50,#1]{\textsf{Antigoni:} #2}}
}{
    \newcommand{\alex}[2][]{}
    \newcommand{\anti}[2][]{}
}

\iftoggle{neurips}{\title{DMM: Distributed Matrix Mechanism for Differentially-Private Federated Learning using Packed Secret Sharing}}{}

\newcommand{\R}{\mathbb{R}}

\newcommand{\N}{\mathbb{N}}
\newcommand{\Z}{\mathbb{Z}}
\newcommand{\F}{\mathbb{F}}

\usepackage{stmaryrd}
\newcommand{\shr}[1]{[ #1 ]}

\newcommand{\packshr}[1]{[ #1 ]}

\newcommand{\block}[1]{\bm{#1}}
\newcommand{\batchpackshares}[2]{\block{#1}_{[1,\ell]}^{#2}}
\newcommand{\mat}[1]{\bm{#1}}

\newcommand{\sens}{\mathrm{sens}_\schema(\mat{C})}
\newcommand{\sensone}{\mathrm{sens}_\schema^1(\mat{C})}
\newcommand{\simplesens}{\mathrm{sens}_\schema}
\newcommand{\numRounds}{T^*}
\newcommand{\schema}{\Phi}
\newcommand{\pattern}{\phi}
\newcommand{\adjacencyDelta}{\mathfrak{D}}
\newcommand{\nbrs}{\mathtt{Nbrs}}

\newcommand{\numParties}{n}
\newcommand{\totalParties}{N}
\newcommand{\Party}{P}
\newcommand{\Committee}[1]{\mathcal{C}_{#1}}
\newcommand{\Dropouts}{\mathsf{Drop}}

\newcommand{\Func}{\mathcal{F}}
\newcommand{\Att}{\mathcal{A}}
\newcommand{\Sim}{\mathcal{S}}
\newcommand{\View}{\mathsf{View}}
\newcommand{\Real}{\mathsf{Real}}
\newcommand{\Ideal}{\mathsf{Ideal}}
\newcommand{\CorrSet}{\mathfrak{C}}
\newcommand{\Output}{\mathsf{Output}}
\newcommand{\SD}{\mathsf{SD}}

\newcommand{\Comm}{\mathsf{Comm}}
\newcommand{\Open}{\mathsf{Open}}

\newcommand{\Exp}{\mathbb{E}}

\newcommand{\Alg}{\mathcal{A}}

\newcounter{protocol}
\newenvironment{protocol}[1]
  {\par\addvspace{\topsep}
   \noindent
   \tabularx{\linewidth}{@{} X @{}}
    \\\hline
    \refstepcounter{protocol}\textbf{Protocol \theprotocol} #1 
    \\\hline}
  {\hrule
  }

\newcommand{\secparam}{{\lambda}}

\newcommand{\Agg}{{\sf Agg}}
\newcommand{\PFL}{{ \Pi_{\sf DMM}}}
\newcommand{\share}{{\sf Share }}
\newcommand{\reshare}{{\sf Reshare }}
\newcommand{\reconstruct}{{\sf Recons}}
\newcommand{\recover}{{\sf Recover }}
\newcommand{\calC}{{\cal C}}

\newcommand{\PSS}{{\sf LRP}}
\newcommand{\SecAgg}{{\sf SecAgg}}
\newcommand{\SecAggEnc}{{\sf SecAgg.Enc}}

\newcommand{\noisedist}{\mathcal{D}}

\begin{document}

\twocolumn[\icmltitle{DMM: Distributed Matrix Mechanism for Differentially-Private Federated Learning Based on Constant-Overhead Linear Secret Resharing}

\begin{icmlauthorlist}
\icmlauthor{Alexander Bienstock}{air}
\icmlauthor{Ujjwal Kumar}{jp}
\icmlauthor{Antigoni Polychroniadou}{air}
\end{icmlauthorlist}

\icmlaffiliation{air}{J.P.~Morgan AI Research \& J.P.~Morgan AlgoCRYPT CoE, New York, New York, USA}
\icmlaffiliation{jp}{J.P.~Morgan, Mumbai, India}

\icmlcorrespondingauthor{Alexander Bienstock}{alex.bienstock@jpmchase.com}

\vskip 0.3in
]

\printAffiliationsAndNotice{}


\iftoggle{abstract}{
\begin{abstract}
    Federated Learning (FL) has gained lots of traction recently, both in industry and academia.
    In FL, a machine learning model is trained using data from various end-users arranged in committees across several rounds.
    Since such data can often be sensitive, a primary challenge in FL is providing privacy while still retaining utility of the model.
    Differential Privacy (DP) has become the main measure of privacy in the FL setting.
    DP comes in two flavors: central and local.
    In the former, a centralized server is trusted to receive the users' raw gradients from a training step, and then perturb their aggregation with some noise before releasing the next version of the model.
    In the latter (more private) setting, noise is applied on users' local devices, and only the aggregation of users' noisy gradients is revealed even to the server.
    Great strides have been made in increasing the privacy-utility trade-off in the central DP setting, by utilizing the so-called \emph{matrix mechanism}.
    However, progress has been mostly stalled in the local DP setting.
    In this work, we introduce the \emph{distributed} matrix mechanism to achieve the best-of-both-worlds; local DP and also better privacy-utility trade-off from the matrix mechanism.
    We accomplish this by proposing a cryptographic protocol that securely transfers sensitive values across rounds\iftoggle{neurips}{, which makes use of \emph{packed secret sharing}.}{.}
    This protocol accommodates the dynamic participation of users per training round required by FL, including those that may drop out from the computation. 
    We provide experiments which show that our mechanism indeed significantly improves the accuracy for standard federated learning tasks by up to 14\% for given privacy levels compared to previous local DP mechanisms, with little added overhead.
\end{abstract}

In Federated Learning (FL), a machine learning model is trained using data from several end-users.
Since such data can often be sensitive, a key challenge in FL is maintaining utility of the trained models, while preserving privacy of the end-users.
FL has experienced an explosion of progress in recent years, both in industry and research.
In terms of use in practice, there have been numerous deployments of FL recently, such as
Google and Apple's privacy-preserving training of machine learning models for making word suggestions in their mobile keyboards~\cite{gboard,apple} and voice assistants~\cite{apple}.
In FL research, new solutions continue to be proposed with better privacy-utility tradeoffs and usability~\cite{treeFL,matrixFL,MultiEpochFL,BandedFL}.

In more detail, FL typically works in a round-based setting, wherein the current model parameters are sent to a set of clients who locally execute a step of Stochastic Gradient Descent on their own data to obtain gradients with respect to a loss function.
These gradients are then aggregated using different techniques to update the model parameters for the next round (e.g.,~\cite{FedAvg,PerFedAvg,FedProx}).
The main privacy metric for FL is differential privacy (DP)~\cite{DP}.
Roughly speaking, DP guarantees that with high probability, one cannot tell whether or not a user participated in a given FL execution.
There are two different notions of DP that can be considered.
In \emph{central} DP, there is a centralized server who receives the gradients from the clients in each round and then updates the model based on its own noisy aggregation of these gradients.
In this case, DP holds with respect to those to whom the server sends the updated models, but not the server itself.
In \emph{local} DP, there may still be a centralized server, however, the users utilize a Secure Aggregation~\cite{DDGFL,SecAgg,PolylogSecAgg,Flamingo} protocol to only release to the server a noisy aggregation of their gradients, and thus DP holds with respect to the server as well.

There has been tremendous progress recently in the area of central DP for FL~\cite{treeFL,matrixFL,MultiEpochFL,BandedFL}.
These works use a sophisticated set of techniques from the DP literature called the \emph{matrix mechanism}~\cite{TreeMech,ContinualDP,MatrixMech} to achieve excellent privacy-utility trade-offs.
Indeed, in this setting, since the central server receives all of the gradients in the clear and samples all noise on its own, it can \emph{correlate} the noise across rounds in a complex manner.
Intuitively, this means that noise can be re-used across rounds without being detected so that the cumulative noise across all rounds is lower compared to sampling new, fresh noise to hide the gradients in each round.

On the other hand, in the setting of local DP, the clients just add noise locally to their gradients, and then these noisy gradients are summed using a Secure Aggregation protocol~\cite{DDGFL,SecAgg,PolylogSecAgg,Flamingo}.
Since the noise is not correlated across epochs via the matrix mechanism like in the central DP setting, the privacy-utility trade-off of local DP is not as good as that of central DP thus far.

In this work, we propose a solution to achieve the ``best-of-both-worlds" of the central and local DP settings, called the \emph{Distributed Matrix Mechanism}.
We achieve privacy with respect to the central server as in the local DP setting, while using correlated noise to get privacy-utility trade-offs close to that of the central DP setting.
To facilitate this, we propose a new cryptographic protocol to efficiently reshare secret information from one committee of users to the next.
While naively resharing values from committee to committee would require a prohibitively expensive $n^2$ overhead, where $n$ is the number of users in each committee (for each secret, each party in a given committee has to reshare their share to all of the parties in the next committee), we use \emph{linear packed secret sharing}~\cite{PackedSS} in a clever way to reduce this overhead to $O(1)$.
As long as if certain corruption conditions are satisfied across all committees (e.g., each committee has $t<(1/2-\varepsilon)n$ corrupted parties), then privacy holds.
In the FL setting, the committees are the sets of users chosen in each training round and the secrets are the noise and gradients from users across many rounds.
With our protocol, we can instantiate the matrix mechanism in a distributed fashion and thus use noise correlated across epochs. 
We do so by taking appropriate linear combinations of shares corresponding to noise and gradients from many rounds, then reconstructing aggregrated gradients with (correlated) noise to the server using Secure Aggregation.

We furthermore maintain privacy even if the corrupted parties are allowed to act maliciously.
To achieve this, we use two main additional ingredients: 
(i) parity check matrices, with which we can catch corrupted parties who do not follow the protocol;
and (ii) random linear combinations, which allow us to perform such checks with low communication overhead.

Another important property that our protocol achieves is \emph{dropout tolerance}.
In FL, the gradients from end-users often comes from mobile devices, and therefore it may not be guaranteed that such users will stay online for the whole round, even if they are honest.
Thus, the protocol must not fail if some (honest) users drop out, while still being able to handle other corrupted users.

We show the utility and concrete costs of our techniques.
We get accuracy improvements of up to 14\% for given privacy levels for the standard Federated EMNIST~\cite{FEMNIST} learning task compared to the prior local DP solution.
Moreover, our techniques modestly increase the communication overhead by just 2.8x compared to the prior local DP solution which is based on Secure Aggregation, while achieving less than 4 seconds of computation per iteration.

}{
\begin{abstract}
    Federated Learning (FL)\@ solutions with central Differential Privacy (DP)\@ have seen large improvements in their utility in recent years arising from the \emph{matrix mechanism}, while FL solutions with distributed (more private) DP have lagged behind.
    In this work, we introduce the \emph{distributed} matrix mechanism to achieve the best-of-both-worlds; better privacy of distributed DP and better utility from the matrix mechanism.
    We accomplish this using a novel cryptographic protocol that securely transfers sensitive values across client committees of different training iterations with constant communication overhead. 
    This protocol accommodates the dynamic participation of users required by FL, including those that may drop out from the computation. 
    We provide experiments which show that our mechanism indeed significantly improves the utility of FL models compared to previous distributed DP mechanisms, with little added overhead.

\end{abstract}

\section{Introduction}
In Federated Learning (FL), a machine learning model is trained using data from several end-users/clients.
Since such data can often be sensitive, a key challenge in FL is maintaining utility of the trained models, while preserving privacy of the end-users.
FL has experienced an explosion of progress in recent years, both in industry and research.
\iftoggle{short}{}{In terms of use in practice, there have been numerous deployments of FL recently, such as
Google's and Apple's privacy-preserving training of machine learning models for making word suggestions in their mobile keyboards~\cite{gboard,apple} and voice assistants~\cite{apple}.
In FL research, new solutions continue to be proposed with better privacy-utility tradeoffs and usability, e.g.,~\cite{BandedFL,toep_single,toep_mult}.}

In more detail, in each training iteration of FL, typically a central server sends the current model parameters to a set of clients, which we call a \emph{committee}, who locally execute a step of Stochastic Gradient Descent on their own data to obtain gradients with respect to a loss function.
These gradients are then aggregated and sent to the central server using different techniques to update the model parameters for the next iteration (e.g.,~\cite{FedAvg,PerFedAvg,FedProx}).

The main privacy metric for FL is differential privacy (DP)~\cite{DP}.
Roughly speaking, DP guarantees that with high probability, one cannot tell whether a user participated in a given FL execution.
There are two different notions of DP that can be considered.
In \emph{central} DP, there is a centralized server who receives the aggregated gradients from the clients in each iteration (perhaps using a Secure Aggregation protocol~\cite{DDGFL,SecAgg,SASH,opa,lerna}) and then updates the model by adding its own DP noise to these aggregated gradients.
See the left side of Figure~\ref{fig:globalandlocal} for a flowchart illustrating the process. 
In this case, DP holds with respect to those to whom the server sends the updated models (assuming that the server did indeed add noise), but not the server itself.
In \emph{distributed} DP, there may still be a centralized server, however, the clients utilize a Secure Aggregation protocol to release to the server \emph{only} an aggregation of their gradients with their own DP noise already added in.
Thus, DP holds with respect to the server as well; in particular, the clients do not need to trust the server to add noise.
Indeed, distributed DP is important for particularly sensitive data that cannot be known by \emph{anyone} else and for which we cannot rely on a central server to protect.
\alex{add local figure if room?}\anti{you mean the citation to the figure?}


\begin{figure*}[t!]
        \begin{subfigure}{0.48\textwidth}
        \centering
        \includegraphics[scale=0.29]{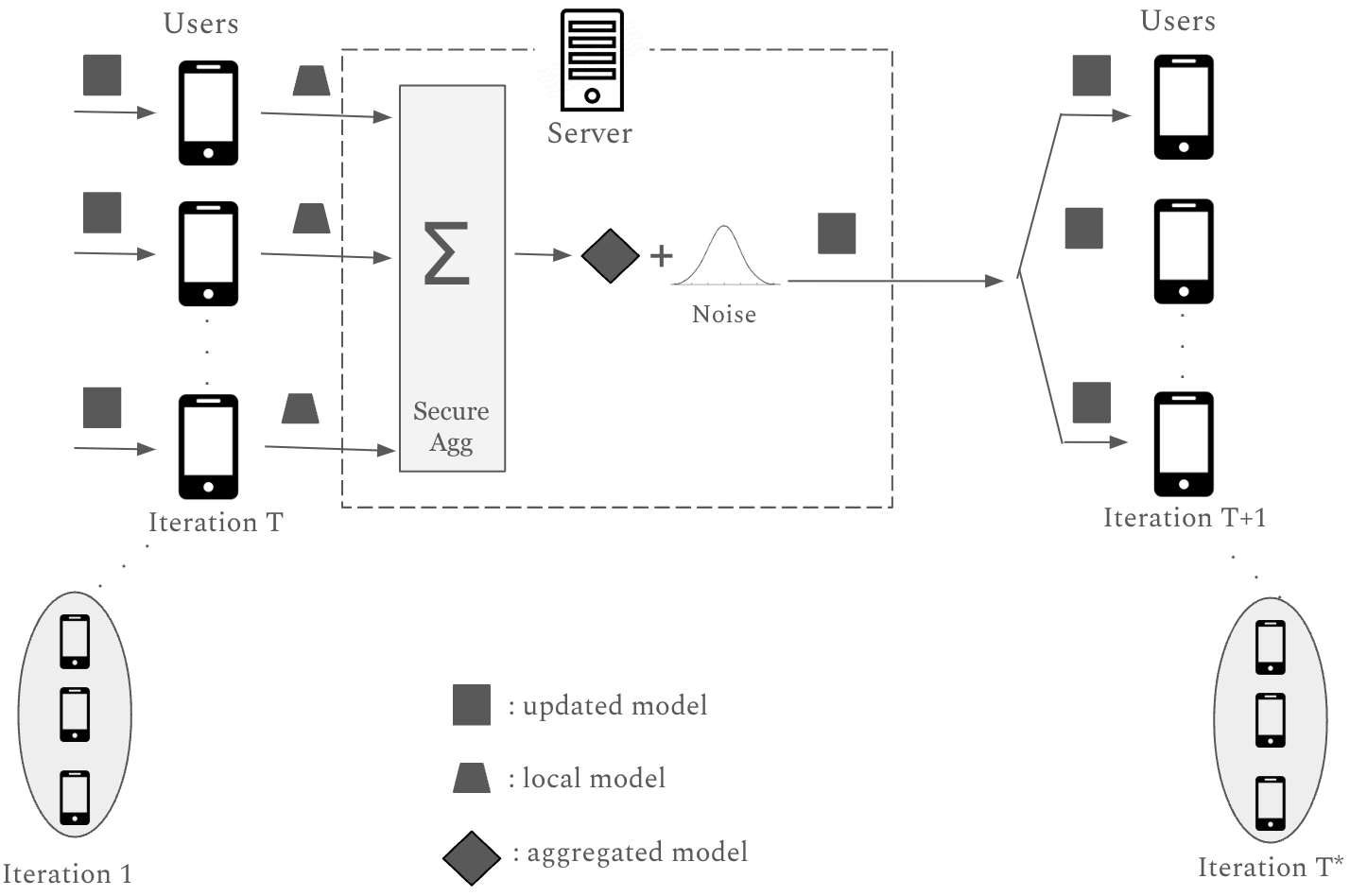}
        \end{subfigure}
        \hfill
        \rule{1pt}{4.5cm}
        \hfill
        \begin{subfigure}{0.48\textwidth}
        \centering
        \includegraphics[width=\textwidth]{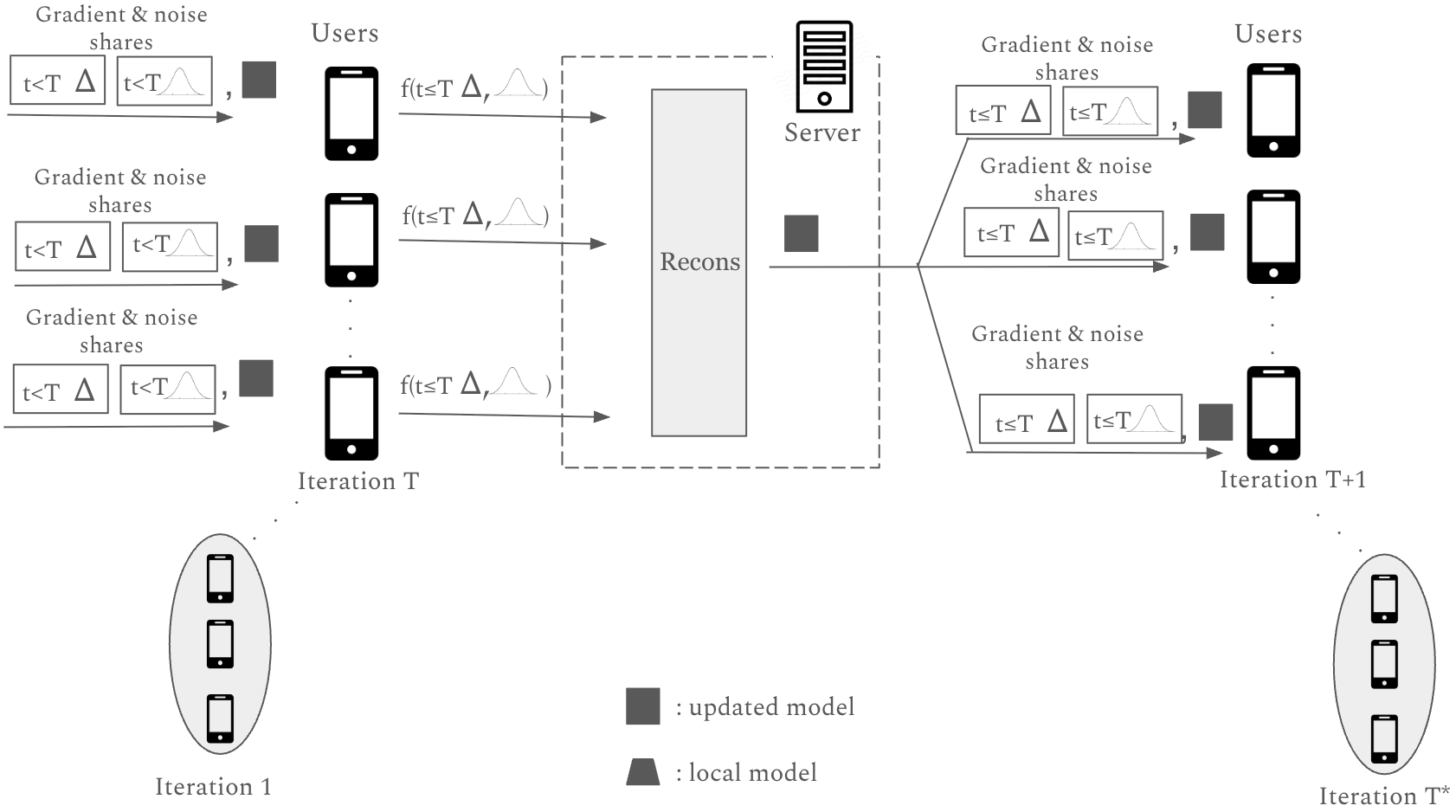}
        \end{subfigure}
        \caption{\textbf{Left}: FL in the central DP model. Users in iteration $T$ update the model locally and these updates are aggregated to the server. The server then adds noise itself before sending the updated model to committee $T+1$.
        \iftoggle{short}{}{DP holds with respect to the other users (assuming the server indeed adds noise), but not the server.}
        \textbf{Right}: FL based on DMM in the distributed DP model.
        Users in iteration $T$ receive noise and gradient shares from previous iterations.
        These parties combine the received shares with shares of their new gradients and freshly sampled noise via a linear combination, $f$, and send these combined shares to the server who uses them to reconstruct (only) the updated model. 
        Afterward, the gradient and noise shares are reshared to the parties in committee 
        $T+1$, ensuring continuity in DMM.
        \iftoggle{short}{}{DP holds with respect to both the other users and the server; moreover, the server is not trusted to add noise.}
        }
    \label{fig:globalandlocal}
    \iffalse}{
    \centering
    \includegraphics[scale=0.35]{figures/global.png}
    \caption{Federated Learning in the central DP model.}
    \label{fig:global}}\fi
\end{figure*}

There has been tremendous progress recently in the area of central DP for FL, e.g.,~\cite{BandedFL,toep_single,toep_mult}.
These works use a sophisticated set of techniques from the DP literature called the \emph{matrix mechanism}~\cite{TreeMech,ContinualDP} to achieve excellent privacy-utility trade-offs.
Indeed, in this setting, since the central server receives all of the gradients in the clear and samples all noise on its own, it can \emph{correlate} the noise across iterations in a complex manner.
Intuitively, this means that noise can be re-used across iterations so that the cumulative noise across all iterations is lower compared to sampling new, fresh noise to hide the gradients in each iteration.

On the other hand, in the setting of distributed DP, the clients just add noise locally to their gradients~\cite{DDGFL,Skell,poisson-bin}. 
Since clients change in each iteration, the noise cannot be correlated across epochs via the matrix mechanism like in the central DP setting, and so the privacy-utility trade-off of distributed DP pales in comparison to that of central DP thus far.

\textbf{Our Contributions.}
In this work, we propose a solution to achieve the ``best-of-both-worlds" of the central and distributed DP settings, 
called the \emph{Distributed Matrix Mechanism} (DMM).
We achieve privacy against the central server, i.e., distributed DP, while using correlated noise to get privacy-utility trade-offs matching the central DP setting.\footnote{We note that, just as in~\cite{DDGFL} and all other works using Secure Aggregation to obtain DP guarantees via aggregated noise, we actually obtain \emph{computational} DP~\cite{compDP}.}\alex{technically we do not experimentally verify this}\anti{arent we running also central DP,or u mean we compare only with the previous local?}

\textbf{1)} DMM starts with \emph{linear secret sharing}~\cite{shamirss}, a central technique in the cryptographic literature which has also been used for FL, e.g.,~\cite{Flamingo,ss-fl-eg-1,ss-fl-eg-2}.
Secret sharing allows for a \emph{dealer} party to distribute to $n$ parties different \emph{shares} of some secret $x$, such that any $t_c$ (corrupted) parties cannot learn anything about $x$ from all of their shares, while any $t_c+k$ parties, for some $k>0$ can use their shares to \emph{reconstruct} $x$.
These shares are also \emph{linear}, meaning that if the users have shares of $x_1$ and $x_2$, they can add their shares together to obtain a sharing of $x_3=x_1+x_2$. 

Typically in FL, secret sharing is used by a single set of parties to secret share (noisy) data to the other parties in the set. 
In our setting, however, 
we additionally need the noise and gradients from users in a given committee to somehow be \emph{reshared} to users in future committees.
Thus, we develop new techniques in this paper to build our \emph{constant-overhead linear secret resharing protocol}, $\PSS$. 
Indeed, our new techniques are paramount, since the naive way to perform such resharing costs $n^2$ communication per secret, instead of $O(1)$ per secret, where $n$ is the number of parties in each committee.
We can see from Table~\ref{tab:soeff} that this results in communication as low as $25.1$ MB per client using our new techniques, and infeasible communication as high as $2.13$ TB per client using the naive resharing.
See Section~\ref{sec:pss} for details on the naive secret resharing protocol and our $\PSS$ protocol with constant communication overhead.

\textbf{2)} Given $\PSS$, we can instantiate the matrix mechanism in a distributed fashion to obtain DMM:
First, the parties take linear combinations of the secret shared gradients and noise, thus introducing noise correlations across epochs.
Then, the parties can reconstruct these aggregrated gradients with (correlated) noise to the server. 
Finally, users (re)share the gradients and noise using $\PSS$.
See the right side of Figure~\ref{fig:globalandlocal} for a flowchart illustrating our approach.

DMM is detailed in Section~\ref{sec:dist-mat-mech}.
Importantly, DMM maintains DP even in the presence of corrupted parties who might manipulate their shares of the gradients and noise.
\iftoggle{short}{}{Indeed, we show that these corrupted parties can only add some independent value $\chi$ to each aggregated noisy gradient received by the server, which can be viewed as a form of post-processing that does not compromise DP.}
Moreover, DMM achieves \emph{dropout tolerance}:
In FL, the gradients from end-users often come from mobile devices, and therefore it may not be guaranteed that such users will stay online for the whole training iteration, even if they are honest.
Thus, the protocol must not fail if some (honest) users drop out\iftoggle{short}{.}{, while still being able to handle other corrupted users.
We design our protocol in a way to be able to still work even if a certain fraction of users drop out in each committee.} 

\textbf{3)} We implement the Distributed Matrix Mechanism using our resharing protocol and empirically test its efficacy in training differentially private FL models.
For example, in Figure~\ref{fig:priv-acc-so}, we show that for Stack Overflow Next Word Prediction~\cite{SO}, our approach improves upon the privacy-utility tradeoff of the most accurate prior distributed DP approach, the Distributed Discrete Gaussian (DDG) Mechanism~\cite{DDGFL}, and matches that of the best central DP approach~\cite{BandedFL}.
We show similar results in Figure~\ref{fig:priv-acc} for Federated EMNIST~\cite{FEMNIST}.
It can also be observed in Tables~\ref{tab:soeff} and~\ref{tab:eff} that our solution is lightweight.
Indeed, DMM adds less than 10 seconds of computation and in some cases less than $2$ MB of communication per client compared to the prior distributed DP approach of secure aggregation.

\begin{table}[]
  \centering
  \footnotesize
  \begin{tabularx}{\iftoggle{neurips}{0.724\columnwidth}{1.01\columnwidth}}{>{\hsize=1.4\hsize}X| p{0.9cm} p{1.0cm} | p{1.09cm} p{1.25cm} p{1.09cm}} 
    \arrayrulecolor{gray!20}\hline
    \rowcolor{gray!20}  & \textbf{$\PSS$ Comp.} & \textbf{$\SecAgg$ Comp.} & \textbf{$\PSS$ Comm.} & \textbf{Naive SR Comm.} & \textbf{$\SecAgg$ Comm.}  \\
    \arrayrulecolor{black}\hline
    Opt. &  7.69 s & 61.3 ms & 4.68 GB & 2.13 TB & 16.2 MB \\
    \hline
    Hon. &  412 ms & 61.3 ms
    &  25.1 MB & 11.4 GB & 16.2 MB \\
  \end{tabularx}
  \caption{Client computation and communication of our $\PSS$ resharing protocol, naive secret resharing, and $\SecAgg$  per training iteration on Stack Overflow Next Word Prediction for committee size $\numParties = 64$.
  We give results for both the optimal~\cite{BandedFL} and more efficient Honaker online~\cite{treeFL,honaker} matrix mechanisms.
  $\SecAgg$ is the bottleneck of prior distributed DP approaches~\cite{DDGFL} and $\PSS$ (as opposed to naive secret resharing) is the bottleneck of DMM.
  (ms $\coloneqq$ milliseconds; s $\coloneqq$ seconds).
  }
    \label{tab:soeff}
\end{table}

\textbf{Related Work.}
In concurrent work,~\cite{EPRINT:BBGKOX24} take another approach to our problem, without secret sharing.
They instead 
separately maintain aggregate noise and gradient encryptions across iterations using (linearly homomorphic) encryption.
These encryptions are then added together by the server and decrypted using clients' noisy secret keys (that do not reveal the actual secret keys) in a clever fasion to reveal \emph{only} the sums with correlated noise in each iteration.
They also sketch a solution with dropout resilience.
Ball et al.~do not provide any code or straightforward method for calculating communication costs; however, we expect their communication complexity to be better than ours.
Yet, since they rely on computational assumptions for linearly homomorphic encryption, whereas we just use information-theoretic secret sharing techinques, we expect ours to be faster.
Moreover, Ball et al.~do not have a maliciously secure protocol.

\begin{figure}[t]
	\centering
	\includegraphics[width=\linewidth]{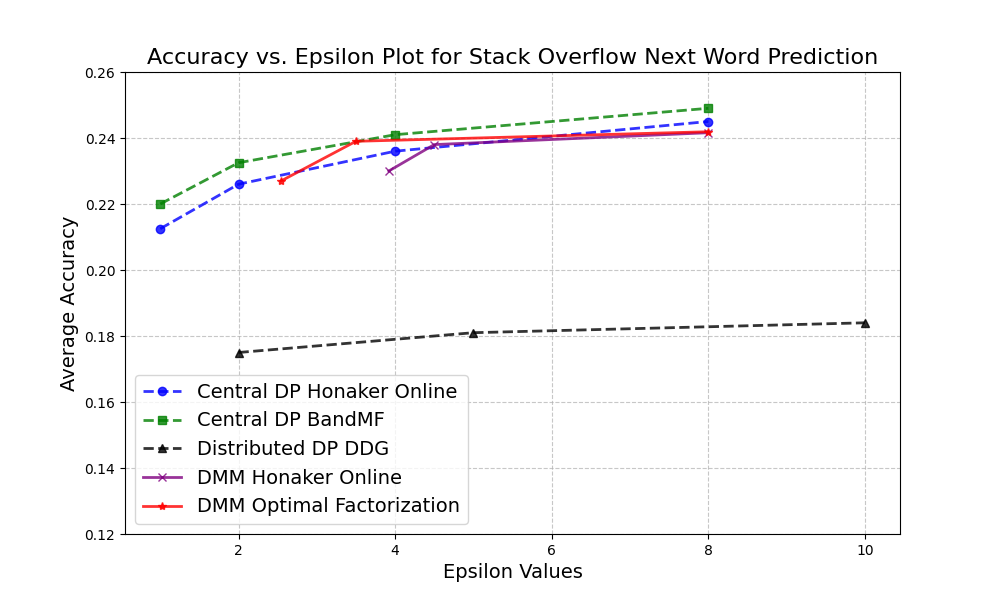}
	\caption{Test accuracies on Stack Overflow Next Word Prediction across different privacy levels $\varepsilon$ for the distributed DP DDG mechanism~\cite{DDGFL}, the central DP BandMF and Honaker online mechanisms~\cite{BandedFL}, and our distributed DP DMM instantiated with the optimal~\cite{BandedFL} and Honaker online~\cite{honaker} matrix factorizations.
		DMM performs 5-6 percentage points better than the prior distributed DP approach and similar (sometimes better) to the prior central DP approaches.
		We use $\delta=1/N$ for $(\varepsilon,\delta)$-DP, where $N$ is the total number of clients selected across training.
	}
	\label{fig:priv-acc-so}
\end{figure}

A fruitful line of works has used correlated noise, and in particular, the matrix mechanism to improve the privacy-utility tradeoff and memory costs of (central) DP FL for increasingly realistic multi-participation settings~\cite{treeFL,matrixFL,MultiEpochFL,BandedFL, toep_single, toep_mult}. 
Privacy amplification techniques like shuffling~\cite{sodaShuffling,focsShuffling} or (Poisson) subsampling~\cite{privacyamp1,poissonZhu,poissonWang} are sometimes used to increase privacy-utility tradeoffs; however, these require strong assumptions on how data is processed which are often not suitable for FL in practice and thus should be avoided~\cite{treeFL}.

New distributed DP mechanisms for FL, the Skellam Mechanism~\cite{Skell}, and Distributed Mean Estimation (DME), the Poisson Binomial Mechanism (PBM)~\cite{poisson-bin}, have appeared in the literature recently, mostly improving the efficiency of DDG.
Indeed, in these works, it is shown that roughly the same privacy-utility tradeoff as DDG is acheived by the Skellam Mechanism, while PBM is not compared empirically to DDG, nor are FL experiments with PBM provided (though they state that their asymptotic error for DME is the same as DDG).
We therefore refer to~\cite{DDGFL} as the state-of-the-art for privacy-utility tradeoff.
Furthermore, PBM specifically does not seem suited to our techniques, since it departs from the \emph{additive noise} paradigm.

Several works have considered so-called \emph{proactive secret sharing}~\cite{proactive,proactiveconstant,churp}.
This setting is similar to ours in which secrets are reshared, however, there the users stay the same in each iteration; just the users that are corrupted changes in each iteration. 
Papers that study a similar model to ours exist, but for more general computations than the special case of aggregation and without a central server that minimizes interaction between clients, and thus are inefficient~\cite{yoso,fluidbienstock,fluidchoudhuri,fluidrachuri,fluidperfect}. \anti{isnt it enough to say that our setting is different and harder, single server setting is harder, why did you say they are inefficient? do you mean inefficnet in the sense that in their setting a secure aggregation protocol will be more efficient?}



\section{Preliminaries}{\label{sec:prelims}}

\subsection{Differentially Private Federated Learning}\label{sec:DPFL-prelim}
In this section, we define some notions important to DPFL. 
Let $\numRounds$ be the number of training iterations, $\numParties$ the number of parties in each committee, and $d$ model dimension.

\textbf{Adjacency and Participation Schemas.}
DP requires a notion of adjacent datasets.
Two data streams $\mat{X}$ and $\tilde{\mat{X}}$ are adjacent if the data associated with any single user is altered, in every iteration in which the user participates.\footnote{We study the more general user-level DP in this work, as opposed to example-level DP.} 
The pattern of when this user participates does not change in these two adjacent streams.
A \emph{participation schema} $\schema$ contains all possible \emph{participation patterns} $\pattern\in\schema$, with each $\pattern\subseteq[\numRounds]$ indicating a set of iterations in which a single user participates.
Let $\nbrs$ be the set of all pairs of neighboring streams $\mat{X}$ and $\adjacencyDelta=\{\mat{X}-\tilde{\mat{X}} : (\mat{X},\tilde{\mat{X}})\in\nbrs\}$ represent the set of all possible differences between neighboring $\mat{X},\tilde{\mat{X}}$.
We say a $\adjacencyDelta$ satisfies the participation schema $\schema$ if the indices of all nonzero rows in each $\R^{\numRounds\times d}$ matrix $\mat{U}\in\adjacencyDelta$ are a subset of some $\pattern\in\schema$.
\iftoggle{short}{}{In this work, we consider the \emph{$b$-min-sep-participation} schema of~\cite{BandedFL}, where any adjacent participations are at least $b$ steps apart.}

\textbf{Centralized DP Matrix Mechanism.}
Let $\mat{A}\in\R^{\numRounds\times \numRounds}$ be an appropriate linear query workload (e.g., prefix sums\iftoggle{short}{)}{ or a matrix encoding of stochastic gradient descent with momentum (SGDM)~\cite{matrixFL})} that is publicly known to all participants.
Matrix mechanisms in the central DP setting use a factorization $\mat{A}=\mat{B}\mat{C}$ to privately estimate the quantity $\mat{A}\mat{X}$ as
$\widehat{\mat{A}\mat{X}}=\mat{B}(\mat{C}\mat{X}+\block{Z}),$
where $\block{Z}$ is sampled by the central server from some noise distribution.

Each entry of the vector $\widehat{\mat{A}\mat{X}}$ corresponds to a model iteration that is released.
The matrix $\mat{A}$ is lower-diagonal, which means that the $T$-th entry of $\widehat{\mat{A}\block{X}}$ only depends on the first $T$ entries of $\block{X}$, for each dimension.
Additionally, the $T$-th entry of $\widehat{\mat{A}\block{X}}$ depends on the first $T$ entries of $\block{Z}$, which means that the noise used in each released model iteration is \emph{correlated}.
\iftoggle{short}{}{This means that each sampled noise element can have \emph{less variance}, resulting in \emph{better accuracy}.}

We now define the \emph{sensitivity} of the central DP matrix mechanism for a particular participation schema $\schema$ with set of neighboring streams $\nbrs$ as 
    $\sens=\sup_{(\mat{X},\tilde{\mat{X}})\in\nbrs} ||\mat{C}\mat{X}-\mat{C}\tilde{\mat{X}}||_F = \sup_{\mat{U}\in\adjacencyDelta} ||\mat{C}\mat{U}||_F$.\footnote{$||\cdot||_F$ is the Frobenius norm.}
As in previous works, it is useful to analyze $\sens$ when all of the contributions from users are clipped to $\ell_2$ norm at most $c=1$, noting that the actual value of $\sens$ scales with $c$ in general.
In our work, however, it is useful to explicitly define the sensitivity for contributions of $\ell_2$ norm $c=1$ as $\sensone$.
The expected total squared error on $\mat{A}$ is typically given as $\mathcal{L}(\mat{B},\mat{C})=\sens ||\mat{B}||_F^2$ and the goal is to find a factorization that minimizes this loss.

\subsection{Problem Statement and Security Model}
For each iteration $T\in[\numRounds]$, we have a committee of (different) clients $\Committee{T}$.
The clients in this committee receive the current model parameters $\theta$ from the server and some values (secret shares)  from the previous committee $\Committee{T-1}$.
Each client $\Party_{T,i}$ uses $\theta$ and their private data to obtain gradients $\block{g}_{T,i}$ and also samples noise $\block{z}_{T,i}$ from some distribution $\noisedist$.
These clients then interact with each other and the server with the goal of revealing \emph{only} $\widehat{\mat{A}\mat{X}}_T=\mat{A}_{[T:,]}\mat{X}+\mat{B}_{[T:,]}\block{Z}$ to the server, where each entry $\block{X}_T = \sum_{i=1}^n \block{g}_{_{T, i}}$ and $\block{Z}_T=\sum_{i=1}^n \block{z}_{_{T, i}}$ for $T\in[\numRounds]$.
We allow $t_c$ clients per iteration as well as the server to be \emph{corrupted} by an adversary $\Att$; i.e., $\Att$ can use the values sent to the corrupted parties to try to learn anything besides each $\widehat{\mat{A}\mat{X}}_T$.
In fact, we handle \emph{malicious} adversaries that can send arbitrary values to other parties.
We allow such adversaries to change each $\widehat{\mat{A}\mat{X}}_T$ received by the server by some additive $\chi_T$ factors that are \emph{independent} of $\block{g}_{_{T, i}}, \block{z}_{_{T, i}}$ for $T\in[\numRounds]$ of the other clients; thus preserving DP.
We also allow $t_d$ honest clients per iteration to drop out; in this case the correct $\widehat{\mat{A}\mat{X}}_T$ should still be received by the server (with added $\chi_T$ defined by the adversary).
We require $t_d+t_c < (1/2-\mu) n$, for constant $0<\mu<1/2$, to guarantee security.
In Section~\ref{sec:reshare_sec}, we formalize this model using a standard simulation-style definition~\cite{GoldBook} and show that such adversaries cannot learn anything besides $\widehat{\mat{A}\mat{X}}$.

\subsection{(Packed) Secret Sharing}\label{sec:packedss}
Let $\F$ be a finite field. 
Recall $t_c$ is the number of maliciously corrupted parties in each committee.
A ($t_c+1$)-out-of-$\numParties$ secret sharing scheme takes as input a secret $z$ from $\F$ and outputs $\numParties$ shares, one for each party, with the property that it is possible to efficiently recover $z$ from every subset of $t_c + 1$ shares, but every subset of at most $t_c$ shares reveals nothing about the secret $z$. \iftoggle{short}{}{The value $t_c$ is called the privacy threshold of the scheme.}

A secret sharing scheme consists of two algorithms: the first, $\share$, takes as input the secret $z$ and the parameters $\numParties$ and $t_c$, and outputs $\numParties$ shares: $(z^1,\dots,z^n)=\share(z, \numParties, t_c)$.
We denote the vector of shares as $\shr{z}_{t_c} = (z^1,\dots,z^n)$.
The second algorithm, $\reconstruct$, takes as input a set of reconstructing parties $\Gamma\subseteq[n]$ and share $z^i$ and outputs a reconstruction value $\reconstruct(\Gamma, z^i)$.
We will utilize secret sharing schemes in which $\lambda_i \cdot z^i=\reconstruct(\Gamma, z^i)$, for some constant $\lambda_i$ dependent on $i$ and $\Gamma$.
If $|\Gamma|\geq t_c+1$, then these reconstruction values can be simply summed to obtain $z=\sum_i\lambda_i\cdot z^i$.
\iftoggle{short}{}{It is required that such a reconstruction of shares generated from a value $z$ reconstructs to the same value $z$.}
The secret sharing scheme we use is also \emph{linear}, meaning that if the parties compute $\shr{z_1}_{t_c}+\shr{z_2}_{t_c}$, then invoke $\reconstruct$ to get reconstruction values $\lambda_i \cdot (z_1^i+z_2^i)$ for a big enough set $\Gamma$ of parties, summing these reconstruction values will yield $z_1+z_2=\sum_i\lambda_i (z_1^i+z_2^i)$.
One instantiation of secret sharing uses a random degree $t_c$ polynomial $f(x)$, where the secret is stored at $f(0)$ and the share of each $\Party_i$ is $f(i)$~\cite{shamirss}.
Reconstruction uses polynomial interpolation, where the $\lambda_i$ are Lagrange coefficients.

\begin{figure*}[t!]
    \centering
    \iftoggle{neurips}{\includegraphics[scale=0.4]{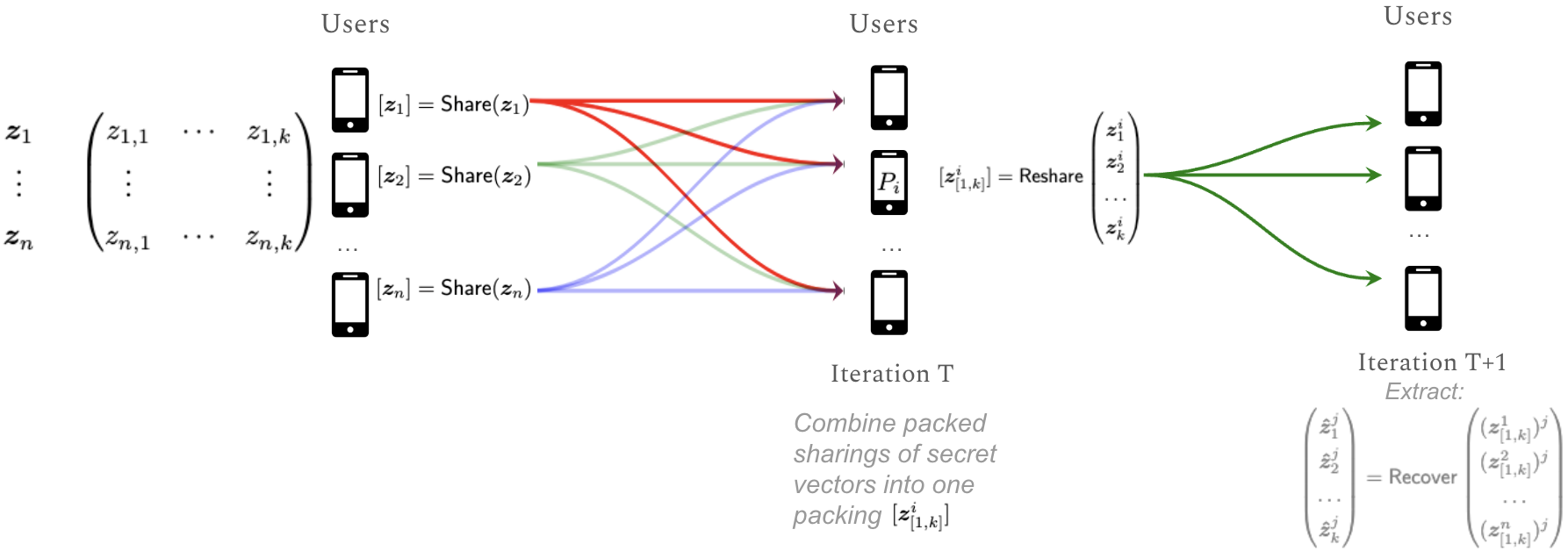}}{\includegraphics[scale=0.52]{figures/sharing_new.png}}
    \caption{Constant-Overhead Linear Secret Resharing Protocol, $\PSS$.
    At a high level, parties $\Party_i$ in iteration $T$ each receive packed secret sharings $\packshr{\block{z}_i}$.
    Then, the parties of iteration $T$ reshare these packed shares by distributing packed sharings \emph{of these packed shares} to the parties of iteration $T+1$, who finally recover packed shares of the original secret vectors $\block{z}_1,\dots,\block{z}_n$.}
    \label{fig:sharing}
\end{figure*}

Packed secret sharing is an extension of secret sharing, where a secret vector $\block{z}=(z_1,\dots,z_k)\in\F^k$ is \emph{packed} into a single set of (individual) shares. 
We call $k$ the \emph{packing parameter}.
\iftoggle{short}{}{This technique is particularly useful for efficiency in cryptographic protocols, as it allows multiple secrets to be shared and reconstructed simultaneously with reduced overhead compared to handling each secret individually.}
We still have that every subset of at most $t_c$ shares reveals nothing about $\block{z}$, but we need at least $t_c+k$ shares to be able to recover $\block{z}$.
There are also similar $\share$ and $\reconstruct$ algorithms, and we denote a sharing of some vector $\block{z}$ as $\packshr{\block{z}}_{t_c+k-1} = (\block{z}^1,\dots,\block{z}^n)$.
In addition, $\reconstruct$ takes as input an index $j\in[k]$ representing the index of the vector to reconstruct.
We utilize packed secret sharing schemes in which $\lambda_i^j \cdot \block{z}^i = \reconstruct(\Gamma, \block{z}^i,j)$, for some constant $\lambda_i^j$. 
If $|\Gamma|\geq t_c + k$, then $z_j$ can be computed as $z_j=\sum_i\lambda_i^j\cdot \block{z}^i$.
The packed secret sharing scheme we use is also \emph{linear} with respect to vector addition of the underlying secrets. 
One instantiation of packed secret sharing extends the polynomial idea from above---$f(x)$ is now degree $t_c+k-1$ and each secret $z_j$ is stored at $f(-j)$;
everything else stays the same~\cite{PackedSS}.

In the following, $t_c$ and $k$ will be fixed, so we will simply refer to packed secret sharings as $\packshr{\block{z}}$.

\section{Linear Secret Resharing Protocol}\label{sec:pss}

In this section, we present our constant-overhead linear secret resharing protocol, $\PSS$.

\textbf{Naive $n^2$-Overhead Protocol.}
We start with the naive $n^2$-overhead protocol, which follows from classical cryptographic literature~\cite{BGW}.
Let it be the case that a secret sharing $\shr{z}$ has been generated for parties in a given committee.
To reshare this value to the parties of the next committee, each party $\Party_i$ in this committee distributes to them a sharing $\shr{z^i}$ of their share.
Since we know it is the case that $z = \sum_i \lambda_i \cdot z^i = \sum_i \reconstruct(\Gamma,z^i)$ and the secret sharing is linear, the parties in the next committee can simply compute their new sharing of $z$ to be $\shr{z'} = \sum_i \lambda_i \cdot \shr{z^i}$, and it is clear that $z' = z$ (the parties of this second committee must also know $\Gamma$, the subset of clients who did not drop out in the first committee).
The problem with this protocol is of course that it has $n^2$ total communication overhead---each of $n$ parties has to distribute $n$ shares to the next committee.

\textbf{Our Constant-Overhead Protocol.}
We can instead start by using packed secret sharing.
Our resharing protocol is pictorially presented and summarized in Figure~\ref{fig:sharing}.
It works by cleverly batching across many packed sharings.
Our resharing protocol consists of four algorithms: it inherits the first algorithm $\share$ from an underlying linear packed secret sharing scheme.
Now, let it be the case that $k$ packed secret sharings $\packshr{\block{z}_1},\dots,\packshr{\block{z}_k}$, for length-$k$ secret vectors $\block{z}_1,\dots,\block{z}_k\in\F^k$, are distributed to the $n$ parties of iteration $T$ (so there are $k^2$ total secrets).
The next algorithm, called the \emph{resharing algorithm}, $\reshare$, takes as input the $k$ packed shares of party $\Party_i$ of iteration $T$, which we denote as the vector $\block{z}^i_{[1,k]}=(\block{z}_1^i,\dots,\block{z}_{k}^i)$, 
and outputs $n$ fresh shares of this vector to the parties of iteration $T+1$: $\packshr{\block{z}^i_{[1,k]}}=((\block{z}^i_{[1,k]})^1,\dots,(\block{z}^i_{[1,k]})^n)=\reshare(\block{z}^i_{[1,k]})$.
Next, the \emph{recovery algorithm}, $\recover$, takes as input the set of dropout parties $\Dropouts_T$ of iteration $T$ and the reshared shares of non-dropout parties of iteration $T$ sent to party $\Party_j$ of iteration $T+1$, $(\block{z}^{i_1}_{[1,k]})^j,\dots,(\block{z}^{i_{\tilde{n}}}_{[1,k]})^j$, for $i_1,\dots,i_{\tilde{n}}\in[n]\backslash\Dropouts_T$, and outputs new shares of the original secret vectors $\block{z}_1,\dots,\block{z}_k$ for party $\Party_j$: $(\hat{\block{z}}_1^j,\dots,\hat{\block{z}}_k^j)=\recover(\Dropouts_T, (\block{z}^{i_1}_{[1,k]})^j,\dots,(\block{z}^{i_{\tilde{n}}}_{[1,k]})^j)$.\footnote{Note: the output shares are for length-$k$ secret vectors $({z}_{1,m},\dots,{z}_{k,m})$ for each $m\in[k]$, instead of $({z}_{\ell,1},\dots,{z}_{\ell,k})$, for each $\ell\in[k]$.} 
The last algorithm $\reconstruct$ is also inherited from the underlying linear packed secret sharing scheme.

We present protocol $\PSS$ below. 
\begin{itemize}[nosep]
    \item $\reshare(\block{z}_{[1,k]}^i)$:\quad Outputs $\packshr{\block{z}_{[1,k]}^i}=\share(\block{z}_{[1,k]}^i)$.
    \item $\recover(\Dropouts, (\block{z}^{i_1}_{[1,k]})^j,\dots,(\block{z}^{i_{\tilde{n}}}_{[1,k]})^j)$:\quad Computes $\hat{\block{z}}_m^j = \sum_l \reconstruct([n]\backslash\Dropouts, (\block{z}^{i_l}_{[1,k]})^j,m)$, for $m\in[k]$. Then outputs $(\hat{\block{z}}_1^j,\dots \hat{\block{z}}_k^j)$
\end{itemize}

Now we observe how $\recover(\cdot)$ outputs packed shares of the original secrets.
Recall that $\reconstruct([n]\backslash\Dropouts, (\block{z}^{i_l}_{[1,k]})^j,m) = \lambda_{i_l}^m\cdot (\block{z}^{i_l}_{[1,k]})^j$, so we can re-write $\hat{\block{z}}_m^j = \sum_l \lambda_{i_l}^m\cdot (\block{z}^{i_l}_{[1,k]})^j $.
Moreover, each $(\block{z}^{i_l}_{[1,k]})^j$ is a share of vector $\block{z}^{i_l}_{[1,k]}=(\block{z}_1^{i_l},\dots,\block{z}_{k}^{i_l})$ for a \emph{linear} packed secret sharing scheme.
Thus, we are computing new packed shares of the vectors $\sum_l \lambda_i^m\cdot (\block{z}_1^{i_l},\dots,\block{z}_{k}^{i_l})$. 
Each $\block{z}_\ell^{i_l}$ was itself $\Party_{i_l}$'s share of vector $\block{z}_\ell$.
Thus the packed shares we are computing indeed share the vectors:
\begin{align*}
    \sum_l & \lambda_{i_l}^m \cdot (\block{z}_1^{i_l},\dots,\block{z}_{k}^{i_l}) \\ & = (\sum_l\reconstruct([n]\backslash\Dropouts, \block{z}_1^{i_l},m),\dots,\\
    & \qquad\qquad\sum_l\reconstruct([n]\backslash\Dropouts, \block{z}_k^{i_l},m)) \\ & = (\block{z}_{1,m},\dots,\block{z}_{k,m}).
\end{align*}

\textbf{Security.} It is clear that the output of $\reshare()$ reveals nothing to the $t_c$ corrupted parties, since it just uses $\share()$ of the underlying packed secret sharing scheme, that is secure against $t_c$ corrupted parties.
Since the number of honest parties that do not dropout is at least $\numParties-t_d -t_c > (1/2+\mu)\numParties$, it is clear that this protocol is resilient to the $t_d$ (honest) dropout parties, if $k\leq 2\mu \numParties$.
This is because $t_c+k\leq (1/2+\mu)\numParties<n-t_d-t_c$, so the shares of the parties that do not dropout can still be used to reconstruct in the secret space during $\recover$.
The fact that malicious parties can only cause reconstructed values to be perturbed by an independent value $\chi$ follows from standard facts about secret sharing~\cite{add-attack}.
We formally prove the security and dropout-resiliency of $\PSS$ in Section~\ref{sec:reshare_sec}.


\textbf{Communication Complexity.}
Let $k=2\mu\numParties$.
The communication cost of both $\share$ and $\reconstruct$ is $\numParties$ field elements for $k$ secrets, which is $1/2\mu$ per secret.
The total communication complexity of $\reshare$ is $\numParties^2$ field elements---each party sends a share to every party in the next iteration.
This is for $k^2=4\mu^2\numParties^2$ secrets, which is $1/4\mu^2$ per secret.
Thus, the communication is at most $1/4\mu^2$ per secret $(\mu < 1/2)$.

\section{Distributed Matrix Mechanism}
\label{sec:dist-mat-mech}

\begin{table*}[t!]
\iftoggle{neurips}{\footnotesize}{}
\begin{protocol}{Differentially-Private Federated Learning Protocol $\PFL$\label{prot:ppfl}}
\textbf{Subprotocols}: $\PSS=(\share,\reshare,\reconstruct,\recover)$ is a secret resharing protocol (See Section~\ref{sec:pss}). 

\textbf{Parameters:} Packing parameter $k\in\N$; number of iterations $\numRounds$; 
finite field $\F$ of bit-width $m$; matrix 
$\mat{A}\in\R^{\numRounds\times\numRounds}$ and $\mat{B},\mat{C}$ such that $\mat{A}=\mat{B}\mat{C}$;
noise distribution $\noisedist$.

\textbf{Inputs:} Current iteration $T$; 
gradients $\block{g}_{T,j,1},\dots,\block{g}_{T,j,k}\in\F^k$ \iftoggle{short}{}{(clipped, scaled, flattened, and rounded/discretized; as in~\cite{DDGFL}; details of this are provided in Section~\ref{sec:discret})}; 
list of dropped clients from iteration $T-1$, $\Dropouts_{T-1}$ (If $|\Dropouts_{T-1}|< t_c+k$, then abort);
reshared gradients and noise received from $\Committee{T-1}$ for the first $T-1$ iterations $\{\shr{\block{X}_{\tau,[1,k]}^{i_l}},\shr{\block{Z}_{\tau,[1,k]}^{i_l}}\}_{\tau\in[T-1],i_l\in[n]\backslash\Dropouts_{T-1}}$. \\

\underline{Round 1:}

\smallskip
\quad\textbf{Parties $\Party_j$:} 
  
\begin{itemize}[noitemsep, topsep=1pt]


\item Sample noise vectors $\block{z}_{_{T,j,1}},\dots,\block{z}_{_{T,j,k}}\in\F^k$ from $\noisedist$. 
\item For each $\ell\in[k]$, distribute packed secret sharings $\packshr{\block{z}_{_{T,j,\ell}}}=\share({\block{z}_{_{T,j,\ell}}})$ and $\packshr{\block{g}_{_{T,j,\ell}}}=\share({\block{g}_{_{T,j,\ell}}})$ to the set $\calC_{T}$ of clients of this training iteration (via authenticated and encrypted channels through the server). 
\end{itemize}

\quad\textbf{Server}:
\begin{itemize}[noitemsep,topsep=1pt]
	\item Receive from each Party $\Party_j$ in $\calC_{T}$ and register dropped clients in list $\Dropouts_T$.
	\item Forward (encrypted) shares $\packshr{\block{z}_{_{T,j,\ell}}}$ and $\packshr{\block{g}_{_{T,j,\ell}}}$ to the other clients of $\calC_{T}$, along with $\Dropouts_T$.
\end{itemize}

\smallskip
\underline{Round 2:}

\smallskip
\quad\textbf{Parties $\Party_j$:} 
\begin{itemize}[noitemsep, topsep=1pt]
\item Receive from server the list of dropped clients from iteration $T$, $\Dropouts_T$.
\item For each $\ell\in[k]$, aggregate $\packshr{\hat{\block{Z}}_{_{T,\ell}}}=(\sum_{\eta\in[\numParties]\setminus\Dropouts_T}\packshr{\block{z}_{_{T,\eta,\ell}}})$ and $\packshr{\hat{\block{X}}_{T,\ell}}=(\sum_{\eta\in[\numParties]\setminus\Dropouts_T}\packshr{\block{g}_{_{T,\eta,\ell}}})$; then reshare $\packshr{\block{Z}_{T,[1,k]}^j}=\reshare(\hat{\block{Z}}_{T,[1,k]}^j)$ and $\packshr{\block{X}_{T,[1,k]}^j}=\reshare(\hat{\block{X}}_{T,[1,k]}^j)$ to the set of clients in $\calC_{T+1}$ (via authenticated and encrypted channels through the sever).
\item{\bf If $T=1$}:
\begin{itemize}[noitemsep]
    \item For each $\ell\in[k]$, compute $\packshr{\block{Y}_{1,\ell}} = \mat{A}_{[1,1]}\cdot \packshr{\hat{\block{X}}_{1,\ell}}+\mat{B}_{[1,1]}\cdot \packshr{\hat{\block{Z}}_{1,\ell}}$, then send shares $\block{Y}_{1,\ell}^j$ to the server. 
\end{itemize}

\item{\bf If $T>1$}:
\begin{itemize}[noitemsep]
    \item For $\tau\in[T-1]$, recover $(\hat{\block{Z}}_{\tau,1}^j,\dots ,\hat{\block{Z}}_{\tau,k}^j) =\recover(\Dropouts_{T-1}, (\block{Z}_{\tau,[1,k]}^{i_1})^j,\dots, (\block{Z}_{\tau,[1,k]}^{i_{\tilde{n}}})^j)$ and $(\hat{\block{X}}_{\tau,1}^j,\dots \hat{\block{X}}_{\tau,k}^j) =\recover(\Dropouts_{T-1}, (\block{X}_{\tau,[1,k]}^{i_1})^j,\dots, (\block{X}_{\tau,[1,k]}^{i_{\tilde{n}}})^j)$ to obtain shares $\packshr{\hat{\block{Z}}_{\tau, \ell}}$ and $\packshr{\hat{\block{X}}_{\tau, \ell}}$ for $\ell\in[k]$.
    
    \item Then again reshare as $\packshr{\block{Z}_{\tau,[1,k]}^j}=\reshare(\hat{\block{Z}}_{\tau,[1,k]}^j)$ and $\packshr{\block{X}_{\tau,[1,k]}^j}=\reshare(\hat{\block{X}}_{\tau,[1,k]}^j)$ to set $\calC_{T+1}$ (via authenticated and encrypted channels through the sever). 
    
    \item For each $\ell\in[k]$, compute  
	$\packshr{\block{Y}_{T,\ell}} = \sum_{\tau=1}^{T} \mat{A}_{[T,\tau]}\cdot \packshr{\hat{\block{X}}_{\tau,\ell}}+\mat{B}_{[T,\tau]}\cdot \packshr{\hat{\block{Z}}_{\tau,\ell}}$, then send shares $\block{Y}_{T,\ell}^j$ to the server. 

\end{itemize}

\end{itemize}

\quad\textbf{Server:}
\begin{itemize}[noitemsep,topsep=1pt]
\item Receive from each Party $\Party_j$ in $\calC_{T}$ and register dropped clients in list $\Dropouts_T$.
\item For each $\ell,m\in[k]$, compute and output $Y_{T,\ell,m} = \sum_{j\in[\numParties]\setminus\Dropouts_T}\reconstruct([\numParties]\backslash\Dropouts_t, \block{Y}_{T,\ell}^j, m)$. 
\item Forward (encrypted) reshares to the clients of $\calC_{T+1}$ and send dropped clients list $\Dropouts_T$ to clients of $\calC_{T+1}$
\vspace{-10pt}
\end{itemize}
\end{protocol}
\end{table*}
We now present our Distributed Matrix Mechanism.
See Protocol~\ref{prot:ppfl} for a detailed description of $\PFL$. 
We will assume that each committee has the same number $n$ of clients.
We write the protocol in terms of batches of gradients of size $k^2$, where $k$ is the packing parameter;
if the model dimension $d>k^2$, then the parties repeat $\PFL$ over batches.\footnote{We also assume that communication between clients in $\PFL$ is done via authenticated and encrypted channels, routed through the server and using a Public-Key Infrastructure (PKI), as in previous works, e.g.,~\cite{SecAgg}, etc.}

Each iteration of this protocol is completed in two communication rounds.\footnote{Note that if there are no dropouts, each iteration can complete in one communication round.}
In the $T$-th iteration, we will assume that the $\numParties$ clients selected have received, 
for $\tau\in[T-1]$, (i) $\shr{\block{X}_{\tau,[1,k]}^{i_l}}$, which are the (aggregated) gradient shares from the first $T-1$ iterations, reshared by non-dropout party $i_l$ in the previous iteration; and
(ii) $\shr{\block{Z}_{\tau,[1,k]}^{i_l}}$, which are the (aggregated) noise shares sampled in the first $T-1 $ iterations, reshared by party $i_l$ in the previous iteration.
The clients first recover shares of the same:
\begin{align*}
	(\hat{\block{Z}}_{\tau,1}^j,\dots ,& \hat{\block{Z}}_{\tau,k}^j) = \\
	& \recover(\Dropouts_{T-1}, (\block{Z}_{\tau,[1,k]}^{i_1})^j,\dots, (\block{Z}_{\tau,[1,k]}^{i_{\tilde{n}}})^j) \\
	(\hat{\block{X}}_{\tau,1}^j,\dots & \hat{\block{X}}_{\tau,k}^j) = \\
	& \recover(\Dropouts_{T-1}, (\block{X}_{\tau,[1,k]}^{i_1})^j,\dots, (\block{X}_{\tau,[1,k]}^{i_{\tilde{n}}})^j),
\end{align*}
based on iteration $T-1$ dropout clients $\Dropouts_{T-1}$ received from the server.

Next, as in the distributed setting, the clients will compute their local gradients $\block{g}_{T,i}$ (clipped, scaled, flattened, and rounded as in~\cite{DDGFL}) 
and sample $\block{z}_{T,i}$ from a noise distribution $\noisedist$.
Then, each client will compute some secret shares $\shr{\block{z}_{T,i}},\shr{\block{g}_{T,i}}$ of their local gradients and noise and distribute them to the other clients of this iteration.
Once receiving these shares, the parties (locally) aggregate them: 
$\shr{\block{Z}_{T}}=(\sum_{\eta=1}^\numParties\shr{\block{z}_{T,\eta}})$ and $\shr{\block{X}_{T}}=(\sum_{\eta=1}^\numParties\shr{\block{g}_{T,\eta}})$.

The parties then take linear combinations, according to $\mat{A}$ and $\mat{B}$, of the packed shares of gradients and noise of all previous iterations, including this one, to obtain shares of the next output of the matrix mechanism, $\packshr{\widehat{\mat{AX}}_T}$.
The parties then reconstruct these noisy gradients to the server (which are then unflattened and rescaled by the server~\cite{DDGFL}).

Finally, the clients will reshare their shares $\hat{\block{Z}}_{\tau,[1,k]}^j$ and $\hat{\block{X}}_{\tau,[1,k]}^j$ of the aggregated noise and gradients from the first $T$ iterations. 
The clients reshare the shares according to protocol $\PSS$ in Section~\ref{sec:pss}.

\textbf{An Optimization.} Typically, the matrix $\mat{A}$ is just the prefix sum matrix\iftoggle{short}{.}{; i.e., a lower triangular matrix of ones, where $\mat{A}_{[i,j]}=1, \forall j\leq i$ and $\mat{A}_{[i,j]}=0,\forall j>i$.}
In this case, the server anyway sees $\widehat{\mat{A}\block{X}}_T-\widehat{\mat{A}\block{X}}_{T-1}=\block{X}_T + (\mat{B}_{[T:,]}-\mat{B}_{[T-1:,]})\block{Z}$, so we can just reconstruct this to it, and do not need to reshare the gradients $\block{X}$ across iterations, saving a factor of two.
This is how we compute the communication complexity of our protocol.

\textbf{Matrix Factorizations  and Communication Complexity.}
We use two different matrix factorizations $\mat{A}=\mat{B}\mat{C}$ for our experiments.
The first is the optimal with respect to the loss function $\mathcal{L}(\mat{B},\mat{C})=\sens ||\mat{B}||_F^2$ for the $b$-min-sep-participation schema $\schema$~\cite{BandedFL}.
In this case, in iteration $T$, the total communication complexity is $(d\cdot T)/(4\mu^2)$, using $k=2\mu \numParties$ as above.
The second factorization is the Honaker Online mechanism~\cite{treeFL,honaker}, where $\mat{C}$ is essentially the binary tree matrix.\footnote{$\PFL$ can easily use any factorization, including BLTs~\cite{toep_mult}, which are not open sourced.}
This mechanism has the benefit that it allows for implementations with only $(d \log T)/(4\mu^2)$ total communication complexity;
in the $T$-th iteration, the released model can be computed by a sum of at most $d\cdot \log(T)$ sharings.

\textbf{Security.} We formally prove the security of $\PFL$ in Section~\ref{sec:reshare_sec} based on the security of $\PSS$, 
i.e., nothing but the noisy gradients are revealed to an adversary corrupting at most $t_c$ parties in each iteration and the server.
We also prove dropout tolerance and distributed DP even in the presence of corrupted parties that perturb the noisy gradients released to the server by independent values $\chi$.


\textbf{Privacy.}
We now state the privacy of our protocol when the noise distribution $\noisedist$ is the Discrete Gaussian distribution with mean $0$ and variance $\sigma^2/\gamma^2$, $\mathcal{N}_\Z(0,\sigma^2/\gamma^2)$.\footnote{Note that, as with prior work in the distributed DP setting, e.g.,~\cite{DDGFL}, we must use discrete noise.
This is because cryptographic techniques work over finite, discrete algebraic structures (like finite fields/rings/groups).}
We note that another option for noise is the Skellam Distribution~\cite{Skell} that yields roughly the same privacy-utility tradeoff; thus we stick to the more standard Discrete Gaussian.
Moreover, our preliminary experiments showed that the Skellam Distribution did not seem to perform even close to as well emprically as the Discrete Gaussian.
Another tempting choice to obtain DP is to use the technique of the Poisson Binomial Mechanism~\cite{poisson-bin}, however, that work departs from the additive noise paradigm and thus does not seem applicable for us.

First we explain some parameters:
$c$ is the norm to which gradients are clipped, $\gamma>0$ is used to determine the granularity for the discretization of gradients, $\beta$ determines the bias of the randomized rounding for discretization, and $\sigma$ is the noise scale of the Discrete Gaussians.
Details on these steps (from~\cite{DDGFL}) are provided in Section~\ref{sec:discret}.
\iftoggle{short}{}{The $\tau$ value in the theorem bounds the max divergence between the sum of $\numParties$ discrete Gaussians each with scale $\sigma/\gamma$ and one discrete Gaussian with scale $\sqrt{n}\sigma/\gamma$.}
The following theorem is proved in Section~\ref{sec:pfs}.
\begin{theorem}\label{thm:privacy}
    Consider a query matrix $\mat{A}\in\R^{\numRounds\times \numRounds}$ along with a fixed factorization $\mat{A}=\mat{B}\mat{C}$ with $\Delta = \sensone$. 
    Let $\tau \coloneqq 10\cdot\sum_{k=1}^{\numParties-1} e^{-2\pi^2\frac{\sigma^2}{\gamma^2}\cdot\frac{k}{k+1}}$ and
    \begin{align*}
        \hat{c}^2 \coloneqq \min \left\{
        \begin{aligned}
            & 
            c^2+\frac{\gamma^2}{4}d+\sqrt{2\log(1/\beta)}\cdot\gamma\cdot(c+\frac{\gamma}{2} \sqrt{d}), \\
            &
            (c+\gamma\sqrt{d})^2
        \end{aligned}
        \right\},
    \end{align*}
    \iftoggle{short}{}{Assume that the number of corruptions in each committee $t_c$ and number of dropouts (of honest parties) in each committee $t_d$ is such that $t_c+t_d< (1/2-\mu)\cdot n$ for $0<\mu<1/2$.}
    Then $\PFL$ satisfies $\frac{1}{2}\varepsilon^2$-concentrated DP for
        \newline$\varepsilon \coloneqq \min \left\{
            \sqrt{\frac{\Delta^2\hat{c}^2}{n\sigma^2} + 2\tau d}, 
            \frac{\Delta\hat{c}}{\sqrt{n}\sigma} + \tau\sqrt{d}
        \right\}$.
\end{theorem}

\textbf{Accuracy.}
We now formally prove the accuracy of our Distributed Matrix Mechanism (DMM).
First, we explain an additional parameter: $m$ is the bit-width of the finite field $\F$ used in 
$\PFL$.
We prove the following in Section~\ref{sec:pfs}.
\begin{theorem}\label{thm:asymp-err}
    Let $\numParties, m, d, \numRounds\in \N$,  and $c,\varepsilon > 0$ satisfy:
    \begin{align*}
        m\geq \tilde{O}\biggl(\max_{T\in[\numRounds]}||\mat{A}_{[T:,]}||_2\sqrt{\numParties T}+\max_{T\in[\numRounds]}||\mat{B}_{[T:,]}||_2\frac{\sqrt{d}\Delta}{\varepsilon}\biggr).
    \end{align*}
    Let $\PFL$ be instantiated with parameters \newline$\gamma = \tilde{O}\biggl(\frac{\max_{T\in[\numRounds]}||\mat{A}_{[T:,]}||_2c\sqrt{\numParties T}}{m \sqrt{d}}+\frac{\max_{T\in[\numRounds]}||\mat{B}_{[T:,]}||_2c\Delta}{\varepsilon m}\biggr)$, $\beta \leq \Theta\left(\frac{1}{\numParties}\right)$, and $\sigma = \tilde{\Theta}\biggl(\frac{c\Delta}{\varepsilon\sqrt{\numParties}}+\sqrt{\frac{d}{\numParties}}\cdot\frac{\gamma\Delta}{\varepsilon}\biggr)$.
    
    Then $\PFL$ attains $\frac{1}{2}\varepsilon^2$-concentrated DP and accuracy:
    \begin{align*}
        \sum_{T=1}^{\numRounds}\Exp & \left[\left|\left| \PFL(X)-  \mat{A}_{[T,:]} \sum_{i=1}^\numParties\mat{X}_i\right|\right|_2^2\right]
        \\ & \leq O\left(||\mat{B}||_F^2\frac{c^2\Delta^2d}{\varepsilon^2}\right).
    \end{align*}
\end{theorem}


\section{Experiments}
Here we empirically evaluate DMM for FL on the Stack Overflow Next Word Prediction (SO-NWP)~\cite{SO} and FEMNIST~\cite{FEMNIST} public benchmarks.
We compare to the prior state-of-the art for privacy-utility tradeoff with distributed DP, the Distributed Discrete Gaussian Mechanism (DDG)~\cite{DDGFL} which also uses privacy amplification via sampling (DMM does not), and central DP, BandMF~\cite{BandedFL}.
Our full experimental setup is described in Section~\ref{sec:moreexps}, and closely follows prior work, including model hyperparameters~\cite{DDGFL,BandedFL}.
All experiments are run on a machine with an AMD EPYC 7R32 processor and an A10G GPU.
See Section~\ref{sec:moreexps} for further evaluation of our results.

\textbf{Privacy Parameters and Selected Hyperparameters.}
For both matrix factorizations, we measure $\sensone$ with respect to the $b$-min-sep-participation schema using~\citet[Theorems 2 and 3]{BandedFL}.
For SO-NWP, we use $\numRounds = 2052$ and $b=342$ (as in~\cite{BandedFL}) and we use $\numRounds=2^{11}=2048$ and $b=512$ for the Honaker factorization, since $\numRounds$ needs to be a power of two.
For FEMNIST, we use $\numRounds = 1445$ (similar to~\cite{DDGFL}) and $b=85$ for the optimal factorization and $\numRounds = 2^{10}=1024$ and $b=64$ for the Honaker factorization---the reason for smaller bands is that there is less data in the FEMNIST dataset, which means clients have to participate more often.

\begin{figure}[t!]
	\centering
	\iftoggle{neurips}{\includegraphics[width=0.6\linewidth]{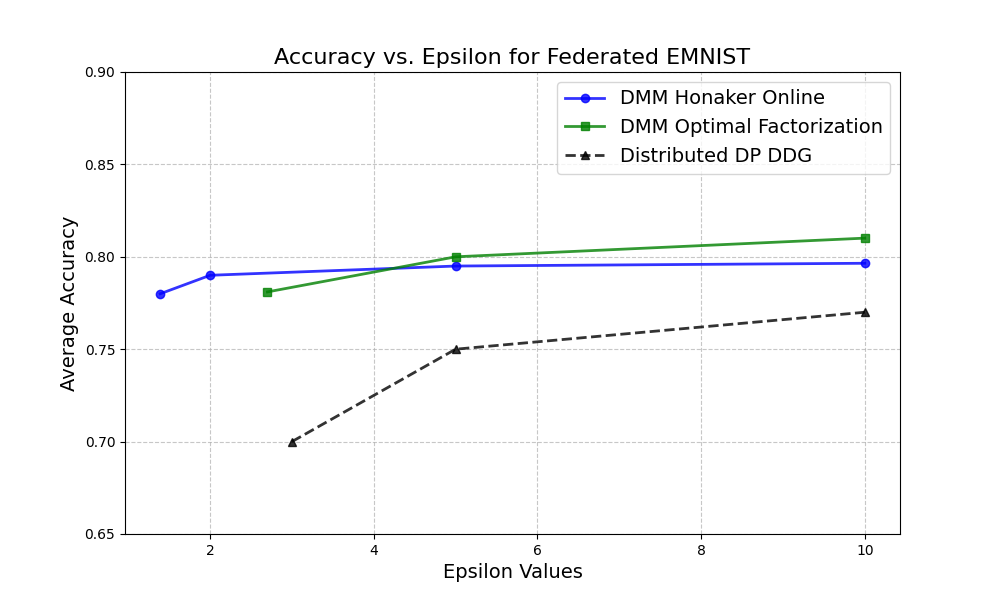}}{\includegraphics[width=\linewidth]{figures/priv_acc4.png}}
	\caption{Test accuracies on FEMNIST across different privacy levels $\varepsilon$ for the distributed DP DDG mechanism and our distributed DP DMM instantiated with the optimal matrix factorization and the Honaker online matrix factorization.
		DMM performs $\approx 4$ percentage points better than the prior distributed DP approach.
		We use $\delta=1/N$ for $(\varepsilon,\delta)$-DP, where $N$ is the total number of clients across training.
	}
	\label{fig:priv-acc}
\end{figure}

\textbf{Results.}
For performance evaluation of DMM, we use $\numParties = 40$ clients per iteration.
Figures~\ref{fig:priv-acc-so} and~\ref{fig:priv-acc} show that across $\varepsilon$ privacy levels, our DMM significantly outperforms DDG in terms of classification accuracy.
This is precisely because DMM uses correlated noise across iterations, whereas DDG uses fresh noise in each iteration.
Our DMM also gets accuracy close to that of the state-of-the-art central DP solutions for SO-NWP.
This gap comes from the error introduced in discretizing values and modular clipping that both arise from DMM representing numbers as finite field elements, as well as some error arising from summing Discrete Gaussians (see Section~\ref{sec:pfs}).
These errors are also present in DDG~\cite{DDGFL}.
We also see that using the Honaker factorization only slightly degrades the accuracy  compared to the mechanism based on the optimal $b$-min-sep-participation matrix factorization.
Therefore, the tree mechanism might be best in practice due to much better efficiency, seen below.

\textbf{Efficiency.}
Tables~\ref{tab:soeff} and~\ref{tab:eff} show client computation and communication costs of DMM for SO-NWP and FEMNIST, respectively, using both the optimal matrix factorization and the Honaker matrix factorization.
We show the costs of $\SecAgg$ (the bottleneck of DDG) using the Flamingo construction~\cite{Flamingo} (recall: any $\SecAgg$ protocol can be used; Flamingo is the state-of-the-art) and the costs of the resharing protocol $\PSS$ in $\PFL$.
We also give the communication costs of the naive secret resharing protocol of Section~\ref{sec:pss}.
For $\PFL$, we assume $\mu=1/6$; i.e., the number of corrupted and dropout parties per iteration satisfies $t_c+t_d < n/3$.
We use 32 bits to represent field values.
For computational experiments, we use $\numParties=64$, as the Flamingo code requires powers of two. 
For the optimal matrix factorization results, we report the worst-case complexity per iteration, which is the penultimate iteration, since clients reshare the noise from all previous iterations.

We see that the naive $n^2$ overhead secret resharing protocol has infeasible communication of up to 2.13 TB per client, which is substantially more than the communication of $\PSS$, with communication as low as 5.73 MB per client.
We also see that the optimal matrix factorization substantially increases both the computation and communication compared to Honaker online factorization.
This suggests that the small increase in accuracy from using the optimal matrix factorization may not be worth it in terms of the added efficiency costs.
Compared to $\SecAgg$ used by DDG, we see a modest increase in computation with the Honaker online factorization from $\PSS$ in $\PFL$;
less than 10 seconds (and sometimes less than 100 ms) per iteration is very reasonable.
In terms of communication, we see an increase of $< 10$ MB with the Honaker online factorization from $\PSS$ in $\PFL$ compared to that of $\SecAgg$ in DDG.
We believe that this added overhead is worth it given the increased accuracy.

\begin{table}[t!]
  \centering
  \footnotesize
  \begin{tabularx}{\iftoggle{neurips}{0.724\columnwidth}{1.01\columnwidth}}{>{\hsize=1.4\hsize}X| p{1.0cm} p{.97cm} | p{1.09cm} p{1.23cm} p{1.09cm}} 
    \arrayrulecolor{gray!20}\hline
    \rowcolor{gray!20}  & \textbf{$\PSS$ Comp.} & \textbf{$\SecAgg$ Comp.} & \textbf{$\PSS$ Comm.} & \textbf{Naive SR Comm.} &\textbf{$\SecAgg$ Comm.}  \\
    \arrayrulecolor{black}\hline
    Opt. &  3.34 s & 58.5 ms & 828 MB & 379 GB & 4.07 MB \\
    \hline
    Hon. &  94.2 ms & 58.5 ms &  5.73 MB & 2.61 GB & 4.07 MB \\
  \end{tabularx}
  \caption{
  Client computation and communication of our $\PSS$ resharing protocol, naive secret resharing, and $\SecAgg$  per iteration on Federated EMNIST for committee size $\numParties = 64$.
  We give results for both the optimal and more efficient Honaker online matrix mechanisms.
  \iftoggle{short}{}{$\SecAgg$ is the bottleneck of DDG and $\PSS$ (as opposed to naive secret resharing) is the bottleneck of DMM.}
  (ms $\coloneqq$ milliseconds; s $\coloneqq$ seconds).
  }
    \label{tab:eff}
\end{table}

\section{Conclusion}
We present in this paper the Distributed Matrix Mechanism (DMM) for FL, which achieves both distributed DP and privacy-utility trade-off of the matrix mechanism for central DP in FL.
Along the way, we introduce a constant-overhead linear secret resharing protocol $\PSS$.
We validate experimentally the utility and efficiency of DMM.
Future work includes designing better low-memory matrix factorizations to get efficiency with better accuracy, as well as adding malicious security to the \emph{encryption} approach of~\cite{EPRINT:BBGKOX24}.

\iftoggle{anonymous}{}{
\section*{Acknowledgements}
We thank Keith Rush for providing valuable assistance with the code used to factorize matrices optimally.
	
\section*{Disclaimer}
Disclaimer: This paper was prepared for informational purposes by the Artificial Intelligence Research group of JPMorgan Chase \& Co. and its affiliates ("JP Morgan'') and is not a product of the Research Department of JP Morgan. JP Morgan makes no representation and warranty whatsoever and disclaims all liability, for the completeness, accuracy or reliability of the information contained herein. This document is not intended as investment research or investment advice, or a recommendation, offer or solicitation for the purchase or sale of any security, financial instrument, financial product or service, or to be used in any way for evaluating the merits of participating in any transaction, and shall not constitute a solicitation under any jurisdiction or to any person, if such solicitation under such jurisdiction or to such person would be unlawful.

© 2025 JPMorgan Chase \& Co. All rights reserved.
}
}

\iftoggle{neurips}{}{
\section*{Impact Statement}
Our work provides privacy for FL using formal $(\varepsilon,\delta)$-DP guarantees.
One should ensure when using DP that the used $(\varepsilon,\delta)$ privacy levels are adequate for protecting sensitive data in their setting.
}

\iftoggle{neurips}{\bibliographystyle{plain}}{\bibliographystyle{icml2025}}
\bibliography{references}

\newpage
\appendix
\onecolumn
{\Large\bfseries\noindent Supplementary Material}
\section{Discretization Details of~\cite{DDGFL}}\label{sec:discret}
We use the randomized rounding strategy from~\cite{DDGFL} for discretization in $\PFL$.
At a high-level, each client first clips and scales their input gradient.
Then, the clients flatten their gradient vectors using some unitary matrix $\mat{U}$ (intuitively, this minimizes the risk of modulo overlap in vector elements that are particularly large).
Finally, the clients use a randomized process to round their gradient vectors in $\R^d$ to $\Z^d$.
On the sever side, after receiving the aggregated, noise outputs $\hat{AX}_T$ in each round, the server unflattens the vector by applying $\mat{U}^T$ and then descales.
Protocols~\ref{prot:clientproc} and~\ref{prot:serverproc} give more detai, but we refer the readers to~\cite{DDGFL} for full details on possible flattening matrices $\mat{U}$ and the randomized rounding procedure used.
\begin{table*}[t]
	\begin{protocol}{Client Gradient Processing\label{prot:clientproc}}
		\textbf{Input}: Gradient $\block{g}_i\in\R^d$.\\
		\textbf{Parameters}: model dimension $d$, clipping threshold $c>0$, granularity $\gamma$, modulus $m$, noise scale $\sigma>0$ and bias $\beta\in[0,1)$.
		\begin{enumerate}
			\item Clip and scale gradient: $\block{g}_i'=\frac{1}{\gamma}\min\{1,\frac{c}{||\block{g}_i||_2}\}\cdot \block{g}_i\in\R^d$.
			\item Flatten vector: $\block{g}_i''= U\cdot\block{g}_i'\in\R^d$.
			\item \textbf{Repeat}:
			\begin{enumerate}
				\item Let $\tilde{\block{g}}_i\in\Z^d$ be a randomized rounding of  $\block{g}_i''$. i.e., $\tilde{\block{g}}_i$ is a product distribution with $\Exp[\tilde{\block{g}}_i] = \block{g}_i''$ and $||\tilde{\block{g}}_i-\block{g}_i''||_\infty<1$.
			\end{enumerate}
			\textbf{until} $|||\tilde{\block{g}}_||_2\leq \min\{c/\gamma +\sqrt{d}, \sqrt{c^2/\gamma^2+\frac{1}{4}d+\sqrt{2\log(1/\beta)}\cdot (c/\gamma+\frac{1}{2}\sqrt{d})}\}$.
			\item \textbf{Output}: $\tilde{\block{g}}_i$.
            \vspace{-10pt}
		\end{enumerate}
	\end{protocol}
\end{table*}

\begin{table*}
	\begin{protocol}{Server Aggregate Noisy Release Value Processing\label{prot:serverproc}}
		\textbf{Input}: Vector $\widehat{AX}_T$.\\
  \textbf{Parameters}: model dimension $d$, clipping threshold $c>0$, granularity $\gamma$, modulus $m$, noise scale $\sigma>0$ and bias $\beta\in[0,1)$.
  \begin{enumerate}
      \item Map $\Z_m$ to $\{1-m/2,2-m/2,\dots,-1,0,1,\dots,m/2-1,m/2\}$ so that $\widehat{AX}_T$ is mapped to $\widehat{AX}_T'\in[-m/2,m/2]^d\cap\Z^d$ (and we have $\widehat{AX}_T'\mod m = \widehat{AX}_T$.
  \end{enumerate}
  \textbf{Output}: $\gamma \cdot U^\intercal \widehat{AX}_T'\in\R^d$.
	\end{protocol}
\end{table*}

To help with the analysis,~\cite{DDGFL} uses the following definitions to represent the conditional randomized rounding.
We present them verabtim.
\begin{definition}[Randomized Rounding]\label{def:randround}
    Let $\gamma>0$ and $d\in\N$.
    Define $R_\gamma:\R^d\to \gamma \Z^d$ (where $\gamma\Z^d\coloneqq \{(\gamma z_1,\gamma z_2,\dots,\gamma z_d): z_1,\dots, z_d\in\Z\}\subseteq\R^d$) as follows.
    For $x\in[0,\gamma]^d$, $R_\gamma(x)$ is a product distribution on $\{0,\gamma\}^d$ with mean $x$;
    that is, independently for each $i\in[d]$, we have $\Pr[R_\gamma(x)_i=0]=1-x_i\gamma$ and $\Pr[R_\gamma(x)_i=\gamma]=x_i/\gamma$.
    In general, for $x\in\R^d$, we have $R_\gamma(x) = \gamma\lfloor x/\gamma\rfloor + R_\gamma(x-\gamma\lfloor x/\gamma\rfloor)$; here $\gamma\lfloor x/\gamma \rfloor \in \gamma\Z^d$ is the point $x$ rounded down coordinate-wise to the grid.
\end{definition}

\begin{definition}[Conditional Randomized Rounding]\label{def:condrandround}
    Let $\gamma>0$ and $d\in \N$ and $G\subseteq \R^d$.
    Define $R_\gamma^G: \R^d\to \gamma \Z^d\cap G$ to be $R_\gamma$ conditioned on the hte output being in $G$.
    That is, $\Pr[R_\gamma^G(x)=y]=\Pr[R\gamma(x)=y]/\Pr[R_\gamma(x)\in G]$ for all $y\in \gamma \Z^d \cap G$, where $R_\gamma$ is as in Definition~\ref{def:randround}.
\end{definition}
\section{Proofs for Section~\ref{sec:dist-mat-mech}}\label{sec:pfs}
\subsection{Proof of Theorem~\ref{thm:privacy}}
First we recall the notion of R\'enyi Divergences and Concentrated Differential Privacy~\cite{concDP1,concDP2}, as well as some other standard DP notions.
We also define the Discrete Gaussian and provide its DP guarantees.
See~\cite{DDGFL} for more details.
Then we prove Thoerem~\ref{thm:privacy}

\begin{definition}[R\'enyi Divergences]
	Let $P$ and $Q$ be probability distributions on some common
	domain $\Omega$. Assume that $P$ is absolutely continuous with respect to $Q$ so that the Radon-Nikodym derivative $P(x)/Q(x)$ is well-defined for $x \in\Omega$.
	
	For $\alpha\in(1,\infty)$, we define the R\'enyi Divergence of order $\alpha$ of $P$ with respect to $Q$ as:
	$$ D_\alpha(P||Q)\coloneqq \frac{1}{\alpha-1}\log \Exp_{X\gets P}\left[ \left(\frac{P(X)}{Q(x)}\right)^{\alpha-1} \right]$$
	We also define
	$$D_*(P||Q)\coloneqq \sup_{\alpha\in(1,\infty)} \frac{1}{\alpha}D_\alpha(P||Q)$$
\end{definition}

\begin{definition}[Concentrated Differential Privacy~\cite{concDP1,concDP2}]
	A randomized algorithm $M:\mathcal{X}^*\to\mathcal{Y}$ satisfies $\frac{1}{2}\varepsilon$-concentrated differential privacy iff, for all $x,x'\in\mathcal{X}$ differing by the addition or removal of a single user's records, we have $D_*(M(x)||M(x'))\leq \frac{1}{2}\varepsilon^2$.
\end{definition}

\begin{definition}[R\'enyi Differential Privacy~\cite{renyidp}]
A randomized algorithm $M:\mathcal{X}^*\to\mathcal{Y}$ satisfies $(\alpha,\varepsilon)$-R\'enyi differential privacy iff, for all $x,x'\in\mathcal{X}$ differing by the addition or removal of a single user's records, we have $D_\alpha(M(x)||M(x'))\leq \frac{1}{2}\varepsilon^2$.
\end{definition}

\begin{definition}[Differential Privacy~\cite{DP}]
A randomized algorithm $M:\mathcal{X}^*\to\mathcal{Y}$ satisfies $(\varepsilon,\delta)$-differential privacy iff, for all $x,x'\in\mathcal{X}$ differing by the addition or removal of a single user's records, we have $$\Pr[M(x)\in E]\leq e^\varepsilon\Pr[M(x')\in E] + \delta,$$
for all events $E\subset Y$.
We refer to $(\varepsilon,0)$-DP as pure DP and $(\varepsilon,\delta)$-DP for $\delta>0$ as approximate DP.
\end{definition}

We remark that $\frac{1}{2}\varepsilon^2$-concentrated DP is equivalent to satisfying $(\alpha, \frac{1}{2}\varepsilon^2\alpha)$-R\'enyi DP simultaneously for all  $\alpha\in(1,\infty)$.
We also have the following conversion lemma from concentrated to approximate DP~\cite{hypteseting,DiscGauss,enhancedGuar}.
\begin{lemma}
	If $M$ satisfies $(\varepsilon,0)$-DP, then it satisfies $\frac{1}{2}\varepsilon^2$-concentrated DP.
	If $M$ satisfies $\frac{1}{2}\varepsilon^2$-DP then, for any $\delta>0$, $M$ satisfies $(\varepsilon_{\mathit{aDP}}(\delta),\delta)$-DP, where
		$$\varepsilon_{\mathit{aDP}}(\delta)=\inf_{\alpha>1}\frac{1}{2}\varepsilon^2\alpha +\frac{\log (1/\alpha\delta)}{\alpha-1}+\log(1-1/\alpha)\leq\varepsilon\cdot(\sqrt{2\log(1/\delta)}+\varepsilon/2).$$
\end{lemma}

\paragraph{Discrete Gaussian}
Here we define the Discrete Gaussiasn~\cite{DiscGauss} and give its DP guarantees.

\begin{definition}[Discrete Gaussian]
	The discrete Gaussian with scale parameter $\sigma>0$ and location parameter $\mu\in\Z$ is a probability distribution supported on the integers $\Z$ denoted by $\mathcal{N}_\Z(\mu,\sigma^2)$ and defined by
	$$\forall x\in \Z \quad \Pr_{X\gets\mathcal{N}_\Z(\mu,\sigma^2)}(X=x) = \frac{\exp\left(\frac{-(x-\mu)^2}{2\sigma^2}\right)}{\sum_{y\in\Z}\exp\left(\frac{-(y-\mu)^2}{2\sigma^2}\right)}.$$
\end{definition}

\begin{proposition}[\cite{DDGFL}, Proposition 14]
	Let $\sigma\geq \frac{1}{2}$.
	Let $X_{I,j}\gets\mathcal{N}_\Z(0,\sigma^2)$ independently for each $i$ and $j$.
	Let $X_i=(X_{i,1},\dots,X_{i,d})\in\Z^d$.
	Let $Z_n=\sum_{i=1}^n X_i\in\Z^d$.
	Then, for all $\Delta\in\Z^d$ and all $\alpha\in(1,\infty)$,
	\begin{align*}
		D_\alpha(Z_n||Z_n+\Delta) \leq & \min\{\frac{\alpha||\Delta||_2^2}{2n\sigma^2}+\tau d, \\ 
		& \frac{\alpha}{2}\cdot \left(\frac{||\Delta||_2^2}{n\sigma^2}+2\frac{||\Delta||_1}{\sqrt{n}\sigma}\cdot\tau+\tau^2d\right),\\
		&  \frac{\alpha}{2}\cdot \left(\frac{||\Delta||_2}{\sqrt{n}\sigma}+\tau\sqrt{d}\right)^2\}
	\end{align*}
	where $\tau\coloneqq10\cdot\sum_{k=1}^n e^{-2\pi^2\sigma^2\frac{k}{k+1}}$. An algorithm $M$ that adds $Z_n$ to a query with $\ell_p$ sensitivity $\Delta_p$ satisfies $\frac{1}{2}\varepsilon^2$-concentrated DP for
	\begin{align*}
		\varepsilon = & \min\{\sqrt{\frac{||\Delta||_2^2}{n\sigma^2}+2\tau d}, \\ 
		& \sqrt{\frac{\Delta_2^2}{n\sigma^2}+2\frac{\Delta_1}{\sqrt{n}\sigma}\cdot\tau+\tau^2d},\\
		&  \frac{\Delta_2}{\sqrt{n}\sigma}+\tau\sqrt{d}\}
	\end{align*}
\end{proposition}

\paragraph{Proof of Theorem~\ref{thm:privacy}}
\begin{proof}
    First, it is sufficient to show that the computation $\mat{C}\mat{G}+\mat{Z}$ satisfies $\frac{1}{2}\varepsilon^2$-concentrated DP, due to the post processing property of DP.
    Now consider two datasets $\mat{G}$ and $\mat{H}$ differing in one data record according to participation schema $\schema$.\footnote{$\mat{G}$ and $\mat{H}$ really consist of entries that are sums of records.}
    By assumption in the theorem statement, we have
    $$\sensone=\Delta, \quad \text{and thus} \quad \sens = c'\cdot\Delta,$$
    where ${c'}$ is the bound on the $\ell_2$ norm of individual gradient vectors that are aggregated.
    Since we use the randomized rounding techniques from Section~\ref{sec:discret}, gradients that are clipped to $\ell_2$ norm $c$ can actually end up having $\ell_2$ norm $c'=\hat{c}$ after rounding, where $\hat{c}$ is as in the theorem statement.
    With the bound on the total sensitivity above, we know from~\citet[Proposition 14]{DDGFL} (reproduced above) that the computation is $\frac{1}{2}\varepsilon^2$-concentrated DP, with the $\varepsilon$ from the theorem statement.
    \end{proof}

\subsection{Proof of Theorem~\ref{thm:asymp-err}}

We first prove the following exact result for the error:
\begin{theorem}\label{thm:error}
    Let $\beta\in[0,1)$, $\sigma^2\geq \gamma/2> 0$, and $c>0$.
    Let $n,d\in\N$ and $\rho \geq 1$.
    Let $\block{g}_{T,i}\in\R^d$ with $||\block{g}_{T,i}||_2\leq c$ for each $T\in[T^*],i\in[n]$.
    Let $U\in\R^{d\times d}$ be a random unitary matrix such that
    $$\forall \block{x}\in\R^d\quad \forall i\in [d] \quad \forall t\in \R\quad  \Exp[\exp(t(Ux)_i)]\leq \exp(t^2\rho||x||_2^2/2d).$$

    Let 
    \begin{align*}
&\Delta=\sensone\\
&\tau=10\cdot\sum_{k=1}^{n-1}e^{-2\pi^2\frac{\sigma^2}{\gamma^2}\cdot\frac{k}{k+1}}\\
&\hat{c}^2 = \min\left\{c^2 + \frac{1}{4}\gamma^2 d + \sqrt{2 \log(1/\beta)}\cdot\gamma\cdot(c+\frac{1}{2}\gamma d), (c+\gamma\sqrt{d})^2\right\}\\
&\varepsilon=\min\left\{\sqrt{\frac{\Delta^2\hat{c}^2}{\numParties \sigma^2}+2\tau d},\frac{\Delta \hat{c}}{\sqrt{\numParties}\sigma} + \tau\sqrt{d}\right\}.
    \end{align*}

    Then $\PFL$ satisfies $\frac{1}{2}\varepsilon^2$-concentrated differential privacy.
    
    Let
    \begin{align*}
        \hat{\sigma}^2(x)  & \coloneqq  \frac{\rho\cdot ||\mat{A}_{[T,:]}||_2^2}{d} \sum_{\tau=1}^T\sum_{i=1}^\numParties ||\block{g}_{\tau,i}||_2^2+\left(\frac{\gamma^2\cdot||\mat{A}_{[T,:]}||_2^2}{4}+\sigma^2\cdot||\mat{B}_{[T,:]}||_2^2\right)\cdot \numParties \\
        & \leq \frac{\rho||\mat{A}_{[T,:]}||_2^2}{d}c^2 \numParties T + \left(\frac{\gamma^2\cdot||\mat{A}_{[T,:]}||_2^2}{4}+||\mat{B}||_2^2\cdot\sigma^2\right)\cdot \numParties 
    \end{align*}
    If $\hat{\sigma}^2(x)\leq r^2$ then 
    \begin{align*}
        \Exp\left[\left|\left|\PFL(x)-\mat{A}_{[T,:]}\left(\sum_{i=1}^\numParties \block{x}_i\right)\right|\right|_2^2\right] & \leq  \frac{d\numParties}{1-\beta}\Biggl(\frac{2\sqrt{2}\cdot r\cdot e^{-r^2/4\hat{\sigma}^2(x)}}{\sqrt{\numParties(1-\beta)^{\numParties T -1}}} \\ & + \biggl(||\mat{A}_{[T,:]}||_2^2\cdot\biggl(\frac{\gamma^2}{4}+\frac{\beta^2\cdot\gamma^2\numParties}{1-\beta}\biggr)  + ||\mat{B}_{[T,:]}||_2^2\cdot\sigma^2\biggr)^{1/2}\Biggr)^2.
    \end{align*}
\end{theorem}

We start with a modified version of Proposition 26 in~\cite{DDGFL}.
\begin{proposition}\label{prop:error_modified}
    Let $R_\gamma^G$ be as in Definition~\ref{def:condrandround} and $G=\{y\in\R^d: ||y||_2^2\leq \Delta^2\hat{c}^2$\}.
    Let $\PFL'(X)$ be $\PFL$ up to the point of modular clipping.
    Consider the parameters from Theorem~\ref{thm:error}.
    Then $\PFL'(X)$ satisfies $\frac{1}{2}\varepsilon^2$-concentrated differential privacy.
    Also the following holds.
    \begin{align*}
        \Exp\left[\left|\left|\PFL'(X)-\mat{A}_{[T,:]}\sum_{i=1}^\numParties \mat{X}_i\right|\right|_2^2\right] \leq ||\mat{A}_{[T,:]}||_2^2\cdot \left(\frac{\gamma^2\cdot d\cdot \numParties}{4(1-\beta)}+\left(\frac{\beta}{1-\beta}\gamma\sqrt{d}\numParties\right)^2\right) + ||\mat{B}_{[T,:]}||_2^2\cdot \numParties \cdot d \cdot \sigma^2.
    \end{align*}
    \begin{align*}
        \forall \block{t}\in\R^d\quad \Exp\left[\exp\left(\left\langle \block{t}, \PFL'(X)-  \mat{A}_{[T,:]}\sum_{i=1}^\numParties \mat{X}_i\right\rangle\right)\right]\leq \frac{\exp((\frac{\gamma^2\cdot||\mat{A}_{[T,:]}||_2^2}{8}+\frac{\sigma^2\cdot||\mat{B}_{[T,:]}||_2^2}{2})\cdot||\block{t}||_2^2\cdot \numParties)}{(1-\beta)^{\numParties T}}.
    \end{align*}
\end{proposition}

\begin{proof}
    First, the differential privacy claim follows from~\citet[Proposition 14]{DDGFL}.
    
    Now, for the utility analysis, we have
    \begin{align*}
        \Exp\left[\left\|\PFL'(X)-\mat{A}_{[T,:]} \sum_{i=1}^\numParties \mat{X}_i\right\|_2^2\right] &= \Exp\left[\left\|\sum_{\tau=1}^T \mat{A}_{T,\tau} \cdot \left(\sum_{i=1}^\numParties (R_\gamma^G(\block{g}_{\tau,i})-\block{g}_{\tau,i})\right) + \mat{B}_{T,\tau}\cdot \sum_{i=1}^\numParties \gamma\cdot \block{z}_{\tau, i}\right\|_2^2\right] \\
        &\leq \sum_{\tau=1}^T\mat{A}_{T,\tau}^2\cdot \Exp\left[\left\|\sum_{i=1}^\numParties R_\gamma^G(\block{g}_{\tau,i})-\block{g}_{\tau,i}\right\|_2^2\right] + \mat{B}_{T,\tau}^2 \cdot \numParties \cdot \sigma^2 \\
        &\leq \left\|\mat{A}_{[T,:]}\right\|_2^2 \cdot \left(\frac{\gamma^2\cdot d\cdot \numParties}{4(1-\beta)}+\left(\frac{\beta}{1-\beta}\gamma\sqrt{d}\numParties\right)^2\right)  + \left\|\mat{B}_{[T,:]}\right\|_2^2 \cdot \numParties\cdot \sigma^2,
    \end{align*}
    where the last inequality is due directly to Proposition 26 of~\cite{DDGFL}.

    Now, for each $i\in[\numParties],\tau\in[T]$, we have that $R_\gamma(\block{g}_{\tau,i})\in \gamma \lfloor \block{g}_{\tau,i}/\gamma\rfloor + \{0,\gamma\}^d$ and is a product distribution with mean $\block{g}_{\tau,i}$.
    Thus, $R_\gamma(\block{g}_{\tau,i}) - \block{g}_{\tau,i} \in \{0,\gamma\}^d$ and is a product distribution with mean $\block{0}$.
    Therefore, by Hoeffding's lemma, we have:
    \begin{align*}
        \forall \block{t}\in\R^d\quad & \Exp[\exp(\langle \block{t}, \sum_{\tau=1}^T \mat{A}_{T,\tau} \sum_{i=1}^{\numParties} R_\gamma(\block{g}_{\tau,i})-\block{g}_{\tau,i}\rangle)] \leq \exp(\frac{\gamma^2}{8}\cdot \numParties\cdot||\mat{A}_{[T,:]}||_2^2\cdot||\block{t}||_2^2).
    \end{align*}
    Thus, 
    \begin{align*}
        \forall \block{t}\in\R^d\quad \Exp[\exp(\langle \block{t}, \sum_{\tau=1}^T \mat{A}_{T,\tau} \sum_{i=1}^{\numParties} R_\gamma^G(\block{g}_{\tau,i})-\block{g}_{\tau,i}\rangle)] & \leq \frac{\Exp[\exp(\langle \block{t}, \sum_{\tau=1}^T \mat{A}_{T,\tau} \sum_{i=1}^{\numParties} R_\gamma(\block{g}_{\tau,i})-\block{g}_{\tau,i}\rangle)]}{\Pr[R_\gamma(\block{g}_{\tau,i})\in G~ \forall \tau,i]} \\
        & \leq \frac{\exp(\frac{\gamma^2}{8}\cdot \numParties\cdot||\mat{A}_{[T,:]}||_2^2\cdot||\block{t}||_2^2)}{(1-\beta)^{n T}}.
    \end{align*}

    Moreover, we have that~\cite{DiscGauss}:
    \begin{align*}
        \forall \block{t} \in \R^d\quad \Exp[\exp(\langle \block{t}, \sum_{\tau=1}^T \mat{B}_{T,\tau}\sum_{i=1}^\numParties \gamma\cdot \block{z}_{\tau,i} \rangle)] \leq  \exp(\frac{\sigma^2}{2}\cdot\numParties \cdot ||\mat{B}_{[T:,]}||_2^2\cdot ||\block{t}||_2^2).
    \end{align*}
\end{proof}
Finally, we are able to prove a modified version of Theorem 36 from~\cite{DDGFL}.
\begin{proof}[Proof of Theorem~\ref{thm:error}]
    First, the differential privacy follows from Proposition~\ref{prop:error_modified} and the post-processing property of DP.

    Now, for the utility, by assumption, we have that
    \begin{align*}
        \forall \block{x}\in \R^d~\forall j\in[d]~\forall t\in\R\quad \Exp[\exp(t(\mat{U}x)_j)]\leq \exp(t^2\rho||\block{x}||_2^2/2d).
    \end{align*}
    Therefore,
    \begin{align*}
        \Exp[\exp(t\cdot (\sum_{\tau=1}^T\mat{A}_{T,\tau}\cdot (\mat{U}\sum_{i=1}^\numParties \block{g}_{\tau,i})_j)] & = \prod_{\tau=1}^T\cdot\prod_{i=1}^\numParties\Exp[\exp(t\cdot \mat{A}_{T,\tau}\cdot (\mat{U}\block{g}_{\tau,i})_j)] \\
        & \leq \prod_{\tau=1}^T\cdot\prod_{i=1}^\numParties\exp(t^2\cdot \mat{A}_{T,\tau}^2\cdot \rho \cdot ||\block{g}_{\tau,i}||_2^2/2d) \\
        & = \exp(t^2\cdot ||\mat{A}_{[T,:]}||_2^2 \cdot \rho \cdot \sum_{\tau=1}^T\sum_{i=1}^\numParties ||\block{g}_{\tau,i}||_2^2/2d).
    \end{align*}
    Combining with the result of Proposition~\ref{prop:error_modified}, we have
    \begin{align*}
        \forall t\in \R~\forall j\in [d]\quad \Exp[\exp(t\cdot (\Alg(\mat{U}x))_j)] & \leq \exp(\frac{t^2\cdot ||\mat{A}_{[T,:]}||_2^2 \cdot \rho}{2d} \cdot \sum_{\tau=1}^T\sum_{i=1}^\numParties ||\block{g}_{\tau,i}||_2^2) \\ & \cdot \frac{\exp((\frac{\gamma^2\cdot||\mat{A}_{[T,:]}||_2^2}{8}+\frac{\sigma^2\cdot||\mat{B}_{[T,:]}||_2^2}{2})\cdot t^2\cdot \numParties)}{(1-\beta)^{\numParties T}}
    \end{align*}
    Recall $\hat{\sigma}^2(x) = \frac{\rho\cdot ||\mat{A}_{[T,:]}||_2^2}{d} \sum_{\tau=1}^T\sum_{i=1}^\numParties ||\block{g}_{\tau,i}||_2^2+(\frac{\gamma^2\cdot||\mat{A}_{[T,:]}||_2^2}{4}+\sigma^2\cdot||\mat{B}_{[T,:]}||_2^2)\cdot \numParties$.

    By Proposition 35 of~\cite{DDGFL}, for all $j\in[d]$,
    \begin{align*}
        \Exp[(M_{[a,b]}(\PFL'(\mat{U}x))_j-\PFL'(\mat{U}x)_j)^2] \leq (b-a)^2\cdot\frac{1}{(1-\beta)^{\numParties T}}\cdot e^{-(b-a)^2/8\hat{\sigma}^2(x)}\cdot(e^{\frac{a^2-b^2}{4\hat{\sigma}^2}}+e^{\frac{b^2-a^2}{4\hat{\sigma^2}}}),
    \end{align*}
    where $a=-r$ and $b=r$ here.
    Summing over $j\in[d]$ gives
    \begin{align*}
        \Exp[||M_{[-r,r]}(\PFL'(\mat{U}x))-\PFL'(\mat{U}x)||_2^2] \leq 4r^2\cdot\frac{d}{(1-\beta)^{\numParties T}}\cdot e^{-r^2/2\hat{\sigma}^2(x)}\cdot 2
    \end{align*}

    Continuing with the proof from~\cite{DDGFL}, we get:
    \begin{align*}
        \Exp&[||{\PFL}(x)-\mat{A}_{[T,:]}\sum_{i=1}\mat{X}_i||_2^2] \\
        &\leq \Biggl((8r^2\cdot\frac{d}{(1-\beta)^{\numParties T}}\cdot e^{-r^2/2\hat{\sigma}^2(x)})^{1/2} + \biggl(||\mat{A}_{[T,:]}||_2^2\cdot\biggl(\frac{\gamma^2\cdot d\cdot \numParties}{4(1-\beta)} +\left(\frac{\beta}{1-\beta}\gamma \sqrt{d}\numParties\right)^2\biggr)  + \\ & ||\mat{B}_{[T,:]}||_2^2\cdot \numParties\cdot d\cdot \sigma^2\biggr)^{1/2}\Biggr)^2 \\
        &= \frac{d\numParties}{1-\beta}\Biggl(\frac{2\sqrt{2}\cdot r\cdot e^{-r^2/4\hat{\sigma}^2(x)}}{\sqrt{\numParties(1-\beta)^{\numParties T -1}}} + \biggl(||\mat{A}_{[T,:]}||_2^2\cdot\biggl(\frac{\gamma^2}{4}+\frac{\beta^2\cdot\gamma^2\numParties}{1-\beta}\biggr)  + ||\mat{B}_{[T,:]}||_2^2\cdot\sigma^2\biggr)^{1/2}\Biggr)^2.
    \end{align*}
\end{proof}

With this error bound, assuming that $\beta\leq 1/\sqrt{n}$ and $\hat{\sigma}^2(x)\leq r^2/4\log(r\sqrt{\numParties}/\gamma^2)$, we get 
\begin{align*}
    \Exp[||\tilde{\Alg}(x)-\mat{A}_{[T,:]}\sum_{i=1}\mat{X}_i||_2^2] \leq O(d\numParties((||\mat{A}_{[T,:]}||_2^2\cdot \gamma^2 + ||\mat{B}_{[T,:]}||_2^2\cdot\sigma^2)).
\end{align*}

\begin{proof}[Proof of Theorem~\ref{thm:asymp-err}]
    Note that $r=\frac{1}{2}\gamma m$.\alex{understand this more?}
    We verify that setting the parameters as specified yields $\frac{1}{2}\varepsilon^2$-concentrated DP and the desired accuracy.
    First, we have that
    \begin{align*}
        \varepsilon^2 \leq \frac{\Delta^2\hat{c}^2}{\numParties\sigma^2} + 2\tau d \leq \frac{\Delta^2(c+\gamma\sqrt{d})^2}{\numParties\sigma^2} + 20\numParties de^{-\pi^2(\sigma/\gamma)^2} \leq \frac{2\Delta^2c^2}{\numParties \sigma^2} + \frac{2d\Delta^2}{\numParties (\sigma/\gamma)^2} + 20\numParties de^{-\pi^2(\sigma/\gamma)^2}.
    \end{align*}
    Thus the privacy requirement is satisfied as long as $\sigma\geq 2c\Delta/\varepsilon\sqrt{\numParties}$ and $(\sigma/\gamma)^2\geq 8d\Delta^2/\varepsilon^2\numParties$, and $20\numParties de^{-\pi^2(\sigma/\gamma)^2}\leq \varepsilon^2/4$.
    So we can set
    \begin{align*}
        \sigma = \max \{\frac{2c\Delta}{\varepsilon\sqrt{\numParties}}, \frac{\gamma \Delta \sqrt{8d}}{\varepsilon\sqrt{\numParties}},\frac{\gamma}{\pi^2}\log(\frac{80\numParties d}{\varepsilon^2})\}=\tilde{\Theta}(\frac{c\Delta}{\varepsilon\sqrt{\numParties}}+\sqrt{\frac{d}{\numParties}}\cdot\frac{\gamma\Delta}{\varepsilon}+\gamma\log(\frac{\numParties d}{\varepsilon^2}).
    \end{align*}

    We set $\beta=\min\{1/n,1/2\}=\Theta(\frac{1}{n})$.

    Next,
    \begin{align*}
        \hat{\sigma}^2 & \leq \frac{\rho ||\mat{A}_{[T,:]}||_2^2}{d}c^2\numParties T+(\frac{\gamma^2||\mat{A}_{[T,:]}||_2^2}{4}+\sigma^2||\mat{B}_{[T,:]}||_2^2)\cdot\numParties \\
        & \leq \frac{\rho ||\mat{A}_{[T,:]}||_2^2}{d}c^2\numParties T+\gamma^2||\mat{A}_{[T,:]}||_2^2\numParties+\sigma^2||\mat{B}_{[T,:]}||_2^2\cdot\numParties \\
        & \leq O(\frac{\rho ||\mat{A}_{[T,:]}||_2^2}{d}c^2\numParties T+\gamma^2||\mat{A}_{[T,:]}||_2^2\numParties + ||\mat{B}_{[T,:]}||_2^2(\frac{c^2\Delta^2}{\varepsilon^2}+\frac{\gamma^2 d \Delta}{\varepsilon^2} + \gamma^2\numParties\log^2(\frac{\numParties d}{\varepsilon^2})) \\
        & \leq O(\frac{\rho ||\mat{A}_{[T,:]}||_2^2}{d}c^2\numParties T+||\mat{B}_{[T,:]}||_2^2\frac{c^2\Delta^2}{\varepsilon^2})) + \gamma^2\cdot O(||\mat{A}_{[T,:]}||_2^2\numParties + ||\mat{B}_{[T,:]}||_2^2(\frac{d \Delta}{\varepsilon^2} + \numParties\log^2(\frac{\numParties d}{\varepsilon^2})).
    \end{align*}

    Now we work out the asymptotics of the accuracy guarantee:
    \begin{align*}
        \Exp&[||\PFL(X)-\mat{A}_{[T,:]}\sum_{i=1}\mat{X}_i||_2^2] \\
        &\leq  \frac{d\numParties}{1-\beta}\Biggl(\frac{2\sqrt{2}\cdot r\cdot e^{-r^2/4\hat{\sigma}^2(x)}}{\sqrt{\numParties(1-\beta)^{\numParties T -1}}} + \biggl(||\mat{A}_{[T,:]}||_2^2\cdot\biggl(\frac{\gamma^2}{4}+\frac{\beta^2\cdot\gamma^2\numParties}{1-\beta}\biggr) + ||\mat{B}_{[T,:]}||_2^2\cdot\sigma^2\biggr)^{1/2}\Biggr)^2. \\
        & \leq O(\numParties d(\frac{re^{-r^2/4\hat{\sigma}^2}}{\sqrt{\numParties}}+\sqrt{||\mat{A}_{[T:,]}||_2^2\gamma^2+||\mat{B}_{[T:,]}||_2^2\sigma^2})) \\
        & \leq O(\numParties d(\frac{r^2e^{-r^2/2\hat{\sigma}^2}}{\numParties}+||\mat{A}_{[T:,]}||_2^2\gamma^2+||\mat{B}_{[T:,]}||_2^2\sigma^2)) \\
        & \leq O(\numParties d(\frac{\gamma^2 m^2}{n}\exp(\frac{-\gamma^2m^2}{8\hat{\sigma}^2})+||\mat{A}_{[T:,]}||_2^2\gamma^2+||\mat{B}_{[T:,]}||_2^2(\frac{c^2\Delta^2}{\varepsilon^2\numParties}+\frac{d\gamma^2\Delta^2}{\varepsilon^2\numParties}+\gamma^2\log^2(\frac{\numParties d}{\varepsilon^2})))) \\
        & \leq O(||\mat{B}_{[T:,]}||_2^2\frac{c^2\Delta^2d}{\varepsilon^2}+\gamma^2\numParties d(\frac{m^2}{\numParties}\exp(\frac{-\gamma^2m^2}{8\hat{\sigma}^2})+||\mat{A}_{[T:,]}||_2^2+||\mat{B}_{[T:,]}||_2^2(\frac{d\Delta^2}{\varepsilon^2\numParties}+\log^2(\frac{\numParties d}{\varepsilon^2}))))
    \end{align*}
    Similarly to the analysis of Theorem 2 in~\cite{DDGFL}, if 
    \begin{align*}
        m^2 &\geq O((||\mat{A_{[T:,]}}||_2^2\numParties+||\mat{B_{[T:,]}}||_2^2(\frac{d\Delta}{\varepsilon^2}+\numParties\log^2(\frac{\numParties d}{\varepsilon^2})))\cdot \log(1+m^2/n) \\
        & = \tilde{O}(||\mat{A_{[T:,]}}||_2^2\numParties+||\mat{B_{[T:,]}}||_2^2(\frac{d\Delta}{\varepsilon^2}+\numParties)),
    \end{align*}
    then we can set
    \begin{align*}
        \gamma^2=O(\frac{\rho||\mat{A}_{[T:,]}||_2^2c^2\numParties T}{d}+\frac{||\mat{B}_{[T:,]}||_2^2c^2\Delta^2}{\varepsilon^2})\cdot \frac{\log(1+m^2/n)}{m^2}
    \end{align*}
    so that $\frac{m^2}{n}\exp(\frac{-\gamma^2m^2}{8\hat{\sigma}}^2)\leq 1$.

    This gives us,
    \iftoggle{neurips}{
    \begin{align*}
        \Exp&[||\tilde{\Alg}(x)-\mat{A}_{[T,:]}\sum_{i=1}\mat{X}_i||_2^2] \\
        &\leq
        O(||\mat{B}_{[T:,]}||_2^2\frac{c^2\Delta^2d}{\varepsilon^2}+\gamma^2\numParties d(1+||\mat{A}_{[T:,]}||_2^2+||\mat{B}_{[T:,]}||_2^2(\frac{d\Delta^2}{\varepsilon^2\numParties}+\log^2(\frac{\numParties d}{\varepsilon^2})))) \\
        & \leq O(||\mat{B}_{[T:,]}||_2^2\frac{c^2\Delta^2d}{\varepsilon^2}+(\frac{\rho||\mat{A}_{[T:,]}||_2^2c^2\numParties T}{d}+\frac{||\mat{B}_{[T:,]}||_2^2c^2\Delta^2}{\varepsilon^2})\cdot \\ & \frac{\log(1+m^2/n)}{m^2}\numParties d(1+||\mat{A}_{[T:,]}||_2^2+||\mat{B}_{[T:,]}||_2^2(\frac{d\Delta^2}{\varepsilon^2\numParties}+\log^2(\frac{\numParties d}{\varepsilon^2})))) \\
        & \leq O(||\mat{B}_{[T:,]}||_2^2\frac{c^2\Delta^2d}{\varepsilon^2}+||\mat{B}_{[T:,]}||_2^2\frac{c^2\Delta^2d}{\varepsilon^2}(\frac{\log(1+m^2/n)}{m^2} \numParties \cdot (\rho||\mat{A}_{[T:,]}||_2^2 T + \\
        & 1+||\mat{A}_{[T:,]}||_2^2+||\mat{B}_{[T:,]}||_2^2(\frac{d\Delta^2}{\varepsilon^2\numParties}+\log^2(\frac{\numParties d}{\varepsilon^2}))))) \\
        & \leq O(||\mat{B}_{[T:,]}||_2^2\frac{c^2\Delta^2d}{\varepsilon^2}(1+\frac{\log(1+m^2/n)}{m^2} \numParties \\ &  \cdot (\rho||\mat{A}_{[T:,]}||_2^2 T + 1+||\mat{A}_{[T:,]}||_2^2+||\mat{B}_{[T:,]}||_2^2(\frac{d\Delta^2}{\varepsilon^2\numParties}+\log^2(\frac{\numParties d}{\varepsilon^2}))))).
    \end{align*}
    }{
    \begin{align*}
        \Exp&[||\tilde{\Alg}(x)-\mat{A}_{[T,:]}\sum_{i=1}\mat{X}_i||_2^2] \\
        &\leq
        O(||\mat{B}_{[T:,]}||_2^2\frac{c^2\Delta^2d}{\varepsilon^2}+\gamma^2\numParties d(1+||\mat{A}_{[T:,]}||_2^2+||\mat{B}_{[T:,]}||_2^2(\frac{d\Delta^2}{\varepsilon^2\numParties}+\log^2(\frac{\numParties d}{\varepsilon^2})))) \\
        & \leq O(||\mat{B}_{[T:,]}||_2^2\frac{c^2\Delta^2d}{\varepsilon^2}+(\frac{\rho||\mat{A}_{[T:,]}||_2^2c^2\numParties T}{d}+\frac{||\mat{B}_{[T:,]}||_2^2c^2\Delta^2}{\varepsilon^2})\\
        &\cdot \frac{\log(1+m^2/n)}{m^2}\numParties d(1+||\mat{A}_{[T:,]}||_2^2+||\mat{B}_{[T:,]}||_2^2(\frac{d\Delta^2}{\varepsilon^2\numParties}+\log^2(\frac{\numParties d}{\varepsilon^2})))) \\
        & \leq O(||\mat{B}_{[T:,]}||_2^2\frac{c^2\Delta^2d}{\varepsilon^2}+||\mat{B}_{[T:,]}||_2^2\frac{c^2\Delta^2d}{\varepsilon^2}(\frac{\log(1+m^2/n)}{m^2} \numParties \\
        & \cdot (\rho||\mat{A}_{[T:,]}||_2^2 T + 1+||\mat{A}_{[T:,]}||_2^2+||\mat{B}_{[T:,]}||_2^2(\frac{d\Delta^2}{\varepsilon^2\numParties}+\log^2(\frac{\numParties d}{\varepsilon^2}))))) \\
        & \leq O(||\mat{B}_{[T:,]}||_2^2\frac{c^2\Delta^2d}{\varepsilon^2}(1+\frac{\log(1+m^2/n)}{m^2} \numParties \cdot (\rho||\mat{A}_{[T:,]}||_2^2 T + 1+||\mat{A}_{[T:,]}||_2^2+||\mat{B}_{[T:,]}||_2^2(\frac{d\Delta^2}{\varepsilon^2\numParties}+\log^2(\frac{\numParties d}{\varepsilon^2}))))).
    \end{align*}}

    So, if
    \begin{align*}
        m^2 & \geq O(\log(1+m^2/n) \numParties \cdot (\rho||\mat{A}_{[T:,]}||_2^2 T + 1+||\mat{A}_{[T:,]}||_2^2+||\mat{B}_{[T:,]}||_2^2(\frac{d\Delta^2}{\varepsilon^2\numParties}+\log^2(\frac{\numParties d}{\varepsilon^2})))) \\
        & = \tilde{O}(\rho||\mat{A}_{[T:,]}||_2^2\numParties T+||\mat{B}_{[T:,]}||_2^2\frac{d\Delta^2}{\varepsilon^2}),
    \end{align*}
    then the mean squared error is $O(||\mat{B}_{[T:,]}||_2^2\frac{c^2\Delta^2d}{\varepsilon^2})$, as required.
    The final bound is obtained by simply summing the above over each round from $T=1$ to $T=\numRounds$.
\end{proof}

\alex[inline]{what is error for just noisy residuals?}
\section{DMM Security Model and Proof}\label{sec:reshare_sec}

\subsection{Security proofs}
We first provide an intuition on the current analysis for proving the security of cryptographic protocols.  In the security proof, we compare between an $n$-party function $f(x_1,\ldots,x_n) = (y_1,\ldots,y_n)$ and a protocol $P(x_1,\ldots,x_n)$ that allegedly privately computes the function $f$. 
Intuitively, a protocol $P$ correctly and privately computes $f$ if the following hold: (a) {\em Correctness:} For every input $\vec{x}=(x_1,\ldots,x_n)$, the output of the parties at the end of the protocol interaction $P$ is the same as $f(\vec{x})$; 
(b) {\em Privacy}: There exists a simulator $\mathcal{S}$ that receives the input and output of the corrupted parties, and can efficiently generate the messages that the corrupted parties received during the protocol execution. The simulator does not know the input/outputs of the honest parties. Intuitively, the fact that the messages sent by the honest parties can be generated from the input/output of the corrupted parties implies that these messages do not contain any additional information about the inputs of the honest parties besides what is revealed from the output of the computation. 

\subsection{Security Model}
We now introduce the formal security model.
We first note that we consider robustness checks on inputs out of the scope of our security model; i.e., we do not cover \emph{poisoning attacks},which have been extensively studied in the literature, e.g.,~\cite{poisoning1,poisoning2}.
Indeed, it is the case that malicious parties can input to the protocol whatever they want as their gradients and noise $\block{g},\block{z}$, which can lead to a meaningless model.

We follow the standard real/ideal world security paradigm of~\cite{GoldBook}.
Consider some multi-party protocol $\Pi$ that is executed by some parties $\Party_1,\dots,\Party_\totalParties$ that are grouped into committees $\Committee{1},\dots,\Committee{T^*}$ from iteration $1$ to iteration $T^*$ and a server $S$.
Note: the committees can be arbitrarily chosen, but our protocol only provides security if the assumption that the number of parties $\Att$ corrupts is at most $t_c$ holds;
in other words, we abstract out the committee selection process.\footnote{In practice, the committee selection is done by the server.}
Each of these parties has inputs $\block{x}_1,\dots,\block{x}_\totalParties$, and they want to evaluate some given \emph{functionality} $\Func$.
In our case, the functionality $\Func_{\mathsf{PPFL}}$ is resharing the inputs from all previous committees to the next committee, in each iteration, and then outputting the $\widehat{\mat{A}\block{X}}_T$ value to the sever in each iteration $T$, given some factorization $\mat{A}=\mat{B}\mat{C}$. 
The security of protocol $\Pi$ is defined by comparing the real-world execution of the protocol with an \emph{ideal}-world evaluation of $\Func$ by a trusted party (ideal functionality), who receives the inputs $\block{x}_1,\dots,\block{x}_\totalParties$ from the parties in the clear and simply sends the relevant parties their outputs $\Func(\block{x}_1,\dots,\block{x}_\totalParties)$ periodically.
There is an adversary $\Att$ that chooses to corrupt at most $t_c$ of the $n$ parties in each iteration, along with the server.
This adversary $\Att$ sees all of the messages and inputs and outputs of the corrupted parties and is allowed to act arbitrarily on their behalf.
Informally, it is required that for every adversary that corrupts some parties during the protocol execution, there is an adversary $\Sim$, also referred to as the \emph{simulator}, which can achieve the same effect and learn the same information in the ideal-world.
This simulator only sees what the corrupted parties send to the honest parties and the outputs, not the inputs $\block{x}$ of the honest parties. 
We now formally describe the security definition.

\paragraph{Real Execution.}
In the real execution, $\Pi$ is executed in the presence of the adversary $\Att$.
The \emph{view} of a party $\Party$ during an execution of $\Pi$, denoted by $\View_\Party^\Pi$ consists of the messages $\Party$ receives from the other parties during the execution and $\Party$'s input.
The execution of $\Pi$ in the presence of $\Att$ on inputs $(\block{x}_1,\dots,\block{x}_\totalParties)$ denoted $\Real_{\Pi,\Att}(\block{x}_1,\dots,\block{x}_\totalParties)$ is defined as $\{\View_{\Party}^\Pi\}_{\Party\in\CorrSet}$.
The output of $\Pi$ to the honest parties in the presence of $\Att$ on inputs $(\block{x}_1,\dots,\block{x}_\totalParties)$ is noted as $\Output$.

\paragraph{Ideal Execution.}
In the ideal execution, the parties and an ideal world adversary $\Sim$ interact with a trusted party (ideal functionality).
The ideal execution proceeds as follows:
As a committee $\Committee{T}$ comes online, the parties $\Party_{T,1},\dots,\Party_{T,\numParties}$ in that committee send their inputs $\block{x}_{T,1},\dots,\block{x}_{T,n}$ to the trusted party, who computes the output $\Func(\block{x}_{1,1},\dots,\block{x}_{T,n})$ to the server for that iteration.
$\Sim$ is also allowed to release a vector $\block{\chi}$, which will be added to the output, to simulate additive attacks.

\begin{definition}
	Protocol $\Pi$ securely computes $\Func$ if for every adversary $\Att$ there exists a simulator $\Sim$ such that 
	$$\SD((\{\View_{\Party}^\Pi\}_{\Party\in\CorrSet}, \Output),(\Sim(\{\block{x}_{T^*,j}\}_{T,j\in\CorrSet(T)},\Func(\block{x}_{1,1},\dots,\block{x}_{T^*,n}),\Func(\block{x}_{1,1},\dots,\block{x}_{T^*,n})+\block{\chi}))\leq \mathsf{negl}(\secparam),\footnote{$\mathsf{negl}(\secparam)$ is any function in $\secparam^{\omega(1)}$}$$
	where $\SD$ is the statistical distance between the two distributions, $\CorrSet(T)$ is the set of corrupted parties in iteration $T$, and $\secparam$ is the \emph{security parameter}.
\end{definition}

\subsection{Security Proof}
We now give the formal security proof.
\begin{theorem}[Security]
	$\PFL$ securely computes $\Func_{\mathsf{PPFL}}$ for $t_c+t_d < (1/2-\mu)n$, $0<\mu<1/2$. 
\end{theorem}

\begin{proof}
We first build the simulator $\Sim$.
The simulator runs $\Att$ internally. 
We describe the simulator for the first iteration $T=1$ and then inductively for the rest.
Throughout, we will (inductively) show that the simulator knows all of the corrupted parties' shares.
We start with the case of a corrupted server $S$.
\paragraph{Corrupted Server}\alex{this could be improved in the future}
In iteration 1, $\Sim$ simulates the shares sent by honest parties of iteration $1$ to corrupted parties of iteration $1$ and $2$ by sampling random values from the field $\F$.
$\Sim$ receives on behalf of the honest parties in committee $\Committee{1}$ the shares sent by corrupted parties from the same committee.
Note that the honest shares completely (and exactly) define these sharings since the number of honest parties is exactly $t_c+k$, and thus $\Sim$ can compute the corrupted parties' shares and the underlying secret ${\block{{y}}}_i$.
With the corrupted parties' shares, along with the output for iteration $T=1$ (which $\Sim$ receives since the server is corrupted), $\Sim$ has $t_c+k$ points on the $(t_c+k-1)$-degree polynomial underlying each output $\packshr{\block{Y}_{1,\ell}}$, and thus can reconstruct the whole sharing.
Based on this, the simulator can send the honest parties' shares to the server.

In subsequent iterations $T>1$, $\Sim$ first simulates the shares of honest parties' new gradients and noise as above, along with the reshares to the next committee.
It also receives the corrupted parties' input shares from iteration $T$,  and can reconstruct the whole sharing $\packshr{\block{y}_i}$, 
in the same way as above.
$\Sim$ also receives on behalf of the honest parties in committee $\Committee{T}$ the reshared shares sent by corrupted parties from iteration $T-1$.
Note that again the honest shares completely (and exactly) define these sharings since the number of honest parties is exactly $t_c+k$, and thus $\Sim$ can compute the corrupted parties' shares as well as the actual underlying reshared shares $\tilde{\block{Z}}_1^i,\dots,\tilde{\block{Z}}_k^i$ of each corrupted party $\Party_i$ in $\Committee{T}$.
Note that these might be different from the actual underlying shares  $\hat{\block{Z}}_1^i,\dots,\hat{\block{Z}}_k^i$ of the corrupted parties which, $\Sim$ knows from above.
Thus, $\Sim$ can compute $\block{e}_m^i\gets\hat{\block{Z}}_{m}^i-\tilde{\block{Z}}_{m}^i$ for each $m\in[k]$.
Since reconstruction used within $\recover$ of $\PSS$ is just a constant $\lambda_i^j$ multiplied by a party's share, the error introduced here is $\lambda_i^j\cdot\block{e}_m^i$ for each $i\in\CorrSet(T)$.
Thus $\Sim$ sets $\block{\chi}_{T}\gets\sum_{i\in \CorrSet(1)} \lambda_i^j\cdot\block{e}_m^i$.
This error will be incorporated in honest parties' shares of $\widehat{\mat{A}\block{X}}$ for all iterations including and after $T$. 
Indeed, in iteration $T$, $\Sim$ uses the computed corrupted parties' shares, along with the output for iteration $T$ with the error $\block{\chi}_\tau$ for all previous iterations $\tau\leq T$ added in, to obtain $t_c+k$ values with which it can reconstruct the whole sharing as above, and thus the shares that honest parties send to the server.

Now we show that this is a good simulation.
By the properties of Shamir Secret Sharing, we know that the at most $t_c$ shares that the adversary receives in the real world for every sharing will be distributed randomly.
Thus the shares that $\Sim$ sends are distributed the same way.
We also showed that the errors $\block{\chi}_T$ and the shares of the corrupted parties are computed exactly as in the real world.
Therefore $\Sim$ perfectly simulates the real world.\footnote{Note that both in the real and ideal world, if an honest party ever hears from less than $t_c+k$ parties from the previous iteration, they will abort since then more parties have dropped out (or have been dropped by the server) than can be tolerated.}

\paragraph{Honest Server}
In the case of an honest server, we can use all of the same simulation above, except instead of simulating the output shares to the server, we can use this same simulation to compute the error values $\block{\chi}_T$ added to the honest output.
We also need to show that, even in the presence of honest dropout parties, if the corrupted parties do not deviate from the protocol description, then the correct values are output to the server.
This is true since the number of honest parties that do not dropout is at least $\numParties-t_d > (1/2+\mu)\numParties$ and $k\leq 2\mu \numParties$.
Indeed, $t_c+k\leq (1/2+\mu)\numParties<n-t_d$, so the shares of the parties that do not dropout can still be used to reconstruct the secrets, both during $\recover$ of $\PSS$ and the actual output reconstruction to the server in each iteration.

This completes the security proof.
\end{proof}
\section{Additional Experimental Results}\label{sec:moreexps}
\begin{figure}[t!]
    \centering
    \includegraphics[width=\linewidth]{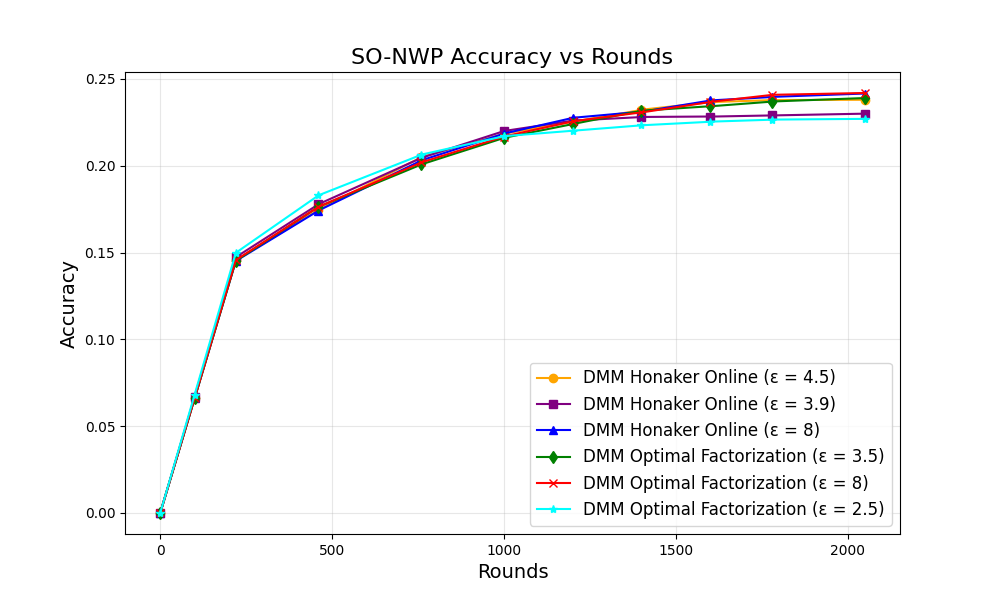}
    \caption{Accuracies during training for SO-NWP across different $\varepsilon$ for both the optimal factorization and Honaker Online facorization}
    \label{fig:so-rounds}
\end{figure}

\textbf{FEMNIST details.} Federated EMNIST is an image classification dataset containing 671,585 training handwritten digit/letter images over 64 classes grouped into $\totalParties = 3400$ clients by their writer.
We use the standard dataset split provided by TensorFlow.
As in~\cite{DDGFL}, we train a small convolutional net with two $3\times 3$ conv layers with $32/64$ channels followed by two fully connected layers with $128/62$ output units;
a $2\times 2$ max pooling layer and two dropout layers with drop rate $0.25/0.5$ are added after the first $3$ trainable layers, respectively.
The total number of parameters is $d=1018174$.
We use namely momentum $0.9$, $1$ client training epoch per iteration, client learning rate $\eta_c=0.02$, server learning rate $\eta_s=1$, and client batch size to $16$.

\textbf{SO-NWP details.}
Stack Overflow is a large-scale text dataset based on the question answering site Stack
Overflow. It contains over 108
training sentences extracted from the site grouped by the $N =
342477$ users, and each sentence has associated metadata such as tags. The task of 
SO-NWP involves predicting
the next words given the preceding words in a sentence
We use the standard dataset split provided by TensorFlow.
As in~\cite{DDGFL,MultiEpochFL}, we use the LSTM architecture defined
in~\cite{fedopt} directly, which has a model size of $d = 4050748$ parameters (slightly
under $2^{22}$).
We use namely momentum $0.9$, $1$ client training epoch per iteration, client learning rate $\eta_c=0.02$, server learning rate $\eta_s=1$, and client batch size to $16$.

\textbf{Accuracy Across Training iterations.}
In Figures~\ref{fig:so-rounds} and~\ref{fig:emnist-rounds}, we show how the accuracies of our different models vary across training iterations.
We show the results for the different matrix factorizations (Honaker Online and optimal) and different privacy values $\varepsilon$.

\begin{figure}[t!]
	\centering
	\includegraphics[width=\linewidth]{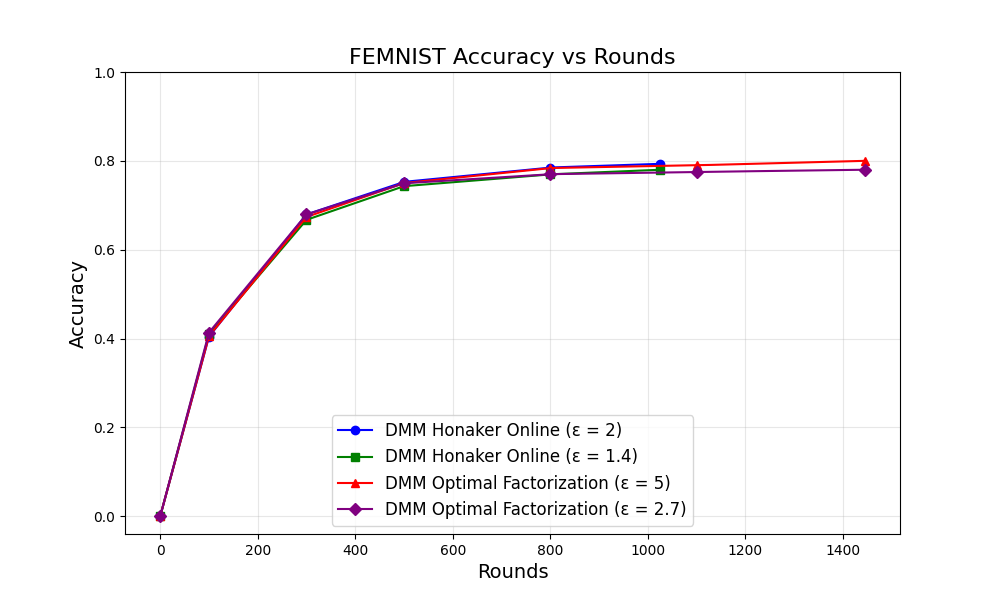}
	\caption{Accuracies during training for FEMNIST across different $\varepsilon$ for both the optimal factorization and Honaker Online facorization}
	\label{fig:emnist-rounds}
\end{figure}

\textbf{Privacy Guarantees with Dropouts and Corrupted Parties.} We note that, just as in~\cite{DDGFL}, our privacy guarantees degrade with corrupt parties and honest dropouts---the amount of combined noise in each iteration is proportional to $(1-\mu')n$ instead of $n$, where $\mu'$ is the fraction of parties that are corrupted or dropped out (recall that we assume $\mu'<(1/2-\mu)$, for $0<\mu<1/2$).
Indeed, the actual obtained $\varepsilon'$ value for DP scales the originally derived $\varepsilon$ value by a $\approx 1/(1-\mu')$ factor.
See~\citet[Figure 9]{DDGFL} for a graphical representation.

\iftoggle{neurips}{

\clearpage
\section*{NeurIPS Paper Checklist}

\begin{enumerate}

\item {\bf Claims}
    \item[] Question: Do the main claims made in the abstract and introduction accurately reflect the paper's contributions and scope?
    \item[] Answer: \answerYes{} 
    \item[] Justification: We provide theorems and experiments that show this is true.
    \item[] Guidelines:
    \begin{itemize}
        \item The answer NA means that the abstract and introduction do not include the claims made in the paper.
        \item The abstract and/or introduction should clearly state the claims made, including the contributions made in the paper and important assumptions and limitations. A No or NA answer to this question will not be perceived well by the reviewers. 
        \item The claims made should match theoretical and experimental results, and reflect how much the results can be expected to generalize to other settings. 
        \item It is fine to include aspirational goals as motivation as long as it is clear that these goals are not attained by the paper. 
    \end{itemize}

\item {\bf Limitations}
    \item[] Question: Does the paper discuss the limitations of the work performed by the authors?
    \item[] Answer: \answerYes{} 
    \item[] Justification: We show that our techniques are somewhat more computationally expensive than the prior work in the local DP setting and suffer some accuraccy loss due to the privacy guarantee.
    \item[] Guidelines:
    \begin{itemize}
        \item The answer NA means that the paper has no limitation while the answer No means that the paper has limitations, but those are not discussed in the paper. 
        \item The authors are encouraged to create a separate "Limitations" section in their paper.
        \item The paper should point out any strong assumptions and how robust the results are to violations of these assumptions (e.g., independence assumptions, noiseless settings, model well-specification, asymptotic approximations only holding locally). The authors should reflect on how these assumptions might be violated in practice and what the implications would be.
        \item The authors should reflect on the scope of the claims made, e.g., if the approach was only tested on a few datasets or with a few runs. In general, empirical results often depend on implicit assumptions, which should be articulated.
        \item The authors should reflect on the factors that influence the performance of the approach. For example, a facial recognition algorithm may perform poorly when image resolution is low or images are taken in low lighting. Or a speech-to-text system might not be used reliably to provide closed captions for online lectures because it fails to handle technical jargon.
        \item The authors should discuss the computational efficiency of the proposed algorithms and how they scale with dataset size.
        \item If applicable, the authors should discuss possible limitations of their approach to address problems of privacy and fairness.
        \item While the authors might fear that complete honesty about limitations might be used by reviewers as grounds for rejection, a worse outcome might be that reviewers discover limitations that aren't acknowledged in the paper. The authors should use their best judgment and recognize that individual actions in favor of transparency play an important role in developing norms that preserve the integrity of the community. Reviewers will be specifically instructed to not penalize honesty concerning limitations.
    \end{itemize}

\item {\bf Theory Assumptions and Proofs}
    \item[] Question: For each theoretical result, does the paper provide the full set of assumptions and a complete (and correct) proof?
    \item[] Answer: \answerYes{} 
    \item[] Justification: Yes we clearly state assumptions and provide complete proofs.
    \item[] Guidelines:
    \begin{itemize}
        \item The answer NA means that the paper does not include theoretical results. 
        \item All the theorems, formulas, and proofs in the paper should be numbered and cross-referenced.
        \item All assumptions should be clearly stated or referenced in the statement of any theorems.
        \item The proofs can either appear in the main paper or the supplemental material, but if they appear in the supplemental material, the authors are encouraged to provide a short proof sketch to provide intuition. 
        \item Inversely, any informal proof provided in the core of the paper should be complemented by formal proofs provided in appendix or supplemental material.
        \item Theorems and Lemmas that the proof relies upon should be properly referenced. 
    \end{itemize}

    \item {\bf Experimental Result Reproducibility}
    \item[] Question: Does the paper fully disclose all the information needed to reproduce the main experimental results of the paper to the extent that it affects the main claims and/or conclusions of the paper (regardless of whether the code and data are provided or not)?
    \item[] Answer: \answerYes{} 
    \item[] Justification: We provide pseudocode for our algorithms and describe datasets, models, and hyperparameters used, as well as our instantiations of the matrix mechanism.
    We also provide the code.
    \item[] Guidelines:
    \begin{itemize}
        \item The answer NA means that the paper does not include experiments.
        \item If the paper includes experiments, a No answer to this question will not be perceived well by the reviewers: Making the paper reproducible is important, regardless of whether the code and data are provided or not.
        \item If the contribution is a dataset and/or model, the authors should describe the steps taken to make their results reproducible or verifiable. 
        \item Depending on the contribution, reproducibility can be accomplished in various ways. For example, if the contribution is a novel architecture, describing the architecture fully might suffice, or if the contribution is a specific model and empirical evaluation, it may be necessary to either make it possible for others to replicate the model with the same dataset, or provide access to the model. In general. releasing code and data is often one good way to accomplish this, but reproducibility can also be provided via detailed instructions for how to replicate the results, access to a hosted model (e.g., in the case of a large language model), releasing of a model checkpoint, or other means that are appropriate to the research performed.
        \item While NeurIPS does not require releasing code, the conference does require all submissions to provide some reasonable avenue for reproducibility, which may depend on the nature of the contribution. For example
        \begin{enumerate}
            \item If the contribution is primarily a new algorithm, the paper should make it clear how to reproduce that algorithm.
            \item If the contribution is primarily a new model architecture, the paper should describe the architecture clearly and fully.
            \item If the contribution is a new model (e.g., a large language model), then there should either be a way to access this model for reproducing the results or a way to reproduce the model (e.g., with an open-source dataset or instructions for how to construct the dataset).
            \item We recognize that reproducibility may be tricky in some cases, in which case authors are welcome to describe the particular way they provide for reproducibility. In the case of closed-source models, it may be that access to the model is limited in some way (e.g., to registered users), but it should be possible for other researchers to have some path to reproducing or verifying the results.
        \end{enumerate}
    \end{itemize}

\item {\bf Open access to data and code}
    \item[] Question: Does the paper provide open access to the data and code, with sufficient instructions to faithfully reproduce the main experimental results, as described in supplemental material?
    \item[] Answer: \answerYes{} 
    \item[] Justification: Yes we include all of the above.
    \item[] Guidelines:
    \begin{itemize}
        \item The answer NA means that paper does not include experiments requiring code.
        \item Please see the NeurIPS code and data submission guidelines (\url{https://nips.cc/public/guides/CodeSubmissionPolicy}) for more details.
        \item While we encourage the release of code and data, we understand that this might not be possible, so “No” is an acceptable answer. Papers cannot be rejected simply for not including code, unless this is central to the contribution (e.g., for a new open-source benchmark).
        \item The instructions should contain the exact command and environment needed to run to reproduce the results. See the NeurIPS code and data submission guidelines (\url{https://nips.cc/public/guides/CodeSubmissionPolicy}) for more details.
        \item The authors should provide instructions on data access and preparation, including how to access the raw data, preprocessed data, intermediate data, and generated data, etc.
        \item The authors should provide scripts to reproduce all experimental results for the new proposed method and baselines. If only a subset of experiments are reproducible, they should state which ones are omitted from the script and why.
        \item At submission time, to preserve anonymity, the authors should release anonymized versions (if applicable).
        \item Providing as much information as possible in supplemental material (appended to the paper) is recommended, but including URLs to data and code is permitted.
    \end{itemize}

\item {\bf Experimental Setting/Details}
    \item[] Question: Does the paper specify all the training and test details (e.g., data splits, hyperparameters, how they were chosen, type of optimizer, etc.) necessary to understand the results?
    \item[] Answer: \answerYes{} 
    \item[] Justification: Yes, as above.
    \item[] Guidelines:
    \begin{itemize}
        \item The answer NA means that the paper does not include experiments.
        \item The experimental setting should be presented in the core of the paper to a level of detail that is necessary to appreciate the results and make sense of them.
        \item The full details can be provided either with the code, in appendix, or as supplemental material.
    \end{itemize}

\item {\bf Experiment Statistical Significance}
    \item[] Question: Does the paper report error bars suitably and correctly defined or other appropriate information about the statistical significance of the experiments?
    \item[] Answer: \answerNo{} 
    \item[] Justification: It would be too computationally expensive to do so. The models took days to train.
    \item[] Guidelines:
    \begin{itemize}
        \item The answer NA means that the paper does not include experiments.
        \item The authors should answer "Yes" if the results are accompanied by error bars, confidence intervals, or statistical significance tests, at least for the experiments that support the main claims of the paper.
        \item The factors of variability that the error bars are capturing should be clearly stated (for example, train/test split, initialization, random drawing of some parameter, or overall run with given experimental conditions).
        \item The method for calculating the error bars should be explained (closed form formula, call to a library function, bootstrap, etc.)
        \item The assumptions made should be given (e.g., Normally distributed errors).
        \item It should be clear whether the error bar is the standard deviation or the standard error of the mean.
        \item It is OK to report 1-sigma error bars, but one should state it. The authors should preferably report a 2-sigma error bar than state that they have a 96\% CI, if the hypothesis of Normality of errors is not verified.
        \item For asymmetric distributions, the authors should be careful not to show in tables or figures symmetric error bars that would yield results that are out of range (e.g. negative error rates).
        \item If error bars are reported in tables or plots, The authors should explain in the text how they were calculated and reference the corresponding figures or tables in the text.
    \end{itemize}

\item {\bf Experiments Compute Resources}
    \item[] Question: For each experiment, does the paper provide sufficient information on the computer resources (type of compute workers, memory, time of execution) needed to reproduce the experiments?
    \item[] Answer: \answerYes{} 
    \item[] Justification: Yes we provide this information
    \item[] Guidelines:
    \begin{itemize}
        \item The answer NA means that the paper does not include experiments.
        \item The paper should indicate the type of compute workers CPU or GPU, internal cluster, or cloud provider, including relevant memory and storage.
        \item The paper should provide the amount of compute required for each of the individual experimental runs as well as estimate the total compute. 
        \item The paper should disclose whether the full research project required more compute than the experiments reported in the paper (e.g., preliminary or failed experiments that didn't make it into the paper). 
    \end{itemize}
    
\item {\bf Code Of Ethics}
    \item[] Question: Does the research conducted in the paper conform, in every respect, with the NeurIPS Code of Ethics \url{https://neurips.cc/public/EthicsGuidelines}?
    \item[] Answer: \answerYes{} 
    \item[] Justification: We conform to the guidelines.
    \item[] Guidelines:
    \begin{itemize}
        \item The answer NA means that the authors have not reviewed the NeurIPS Code of Ethics.
        \item If the authors answer No, they should explain the special circumstances that require a deviation from the Code of Ethics.
        \item The authors should make sure to preserve anonymity (e.g., if there is a special consideration due to laws or regulations in their jurisdiction).
    \end{itemize}

\item {\bf Broader Impacts}
    \item[] Question: Does the paper discuss both potential positive societal impacts and negative societal impacts of the work performed?
    \item[] Answer: \answerYes{} 
    \item[] Justification: We acknowledge that a tradeoff between privacy and accuracy exists
    \item[] Guidelines:
    \begin{itemize}
        \item The answer NA means that there is no societal impact of the work performed.
        \item If the authors answer NA or No, they should explain why their work has no societal impact or why the paper does not address societal impact.
        \item Examples of negative societal impacts include potential malicious or unintended uses (e.g., disinformation, generating fake profiles, surveillance), fairness considerations (e.g., deployment of technologies that could make decisions that unfairly impact specific groups), privacy considerations, and security considerations.
        \item The conference expects that many papers will be foundational research and not tied to particular applications, let alone deployments. However, if there is a direct path to any negative applications, the authors should point it out. For example, it is legitimate to point out that an improvement in the quality of generative models could be used to generate deepfakes for disinformation. On the other hand, it is not needed to point out that a generic algorithm for optimizing neural networks could enable people to train models that generate Deepfakes faster.
        \item The authors should consider possible harms that could arise when the technology is being used as intended and functioning correctly, harms that could arise when the technology is being used as intended but gives incorrect results, and harms following from (intentional or unintentional) misuse of the technology.
        \item If there are negative societal impacts, the authors could also discuss possible mitigation strategies (e.g., gated release of models, providing defenses in addition to attacks, mechanisms for monitoring misuse, mechanisms to monitor how a system learns from feedback over time, improving the efficiency and accessibility of ML).
    \end{itemize}
    
\item {\bf Safeguards}
    \item[] Question: Does the paper describe safeguards that have been put in place for responsible release of data or models that have a high risk for misuse (e.g., pretrained language models, image generators, or scraped datasets)?
    \item[] Answer: \answerNA{} 
    \item[] Justification: We use public datasets
    \item[] Guidelines:
    \begin{itemize}
        \item The answer NA means that the paper poses no such risks.
        \item Released models that have a high risk for misuse or dual-use should be released with necessary safeguards to allow for controlled use of the model, for example by requiring that users adhere to usage guidelines or restrictions to access the model or implementing safety filters. 
        \item Datasets that have been scraped from the Internet could pose safety risks. The authors should describe how they avoided releasing unsafe images.
        \item We recognize that providing effective safeguards is challenging, and many papers do not require this, but we encourage authors to take this into account and make a best faith effort.
    \end{itemize}

\item {\bf Licenses for existing assets}
    \item[] Question: Are the creators or original owners of assets (e.g., code, data, models), used in the paper, properly credited and are the license and terms of use explicitly mentioned and properly respected?
    \item[] Answer: \answerYes{} 
    \item[] Justification: We point out the datasets that we use.
    \item[] Guidelines:
    \begin{itemize}
        \item The answer NA means that the paper does not use existing assets.
        \item The authors should cite the original paper that produced the code package or dataset.
        \item The authors should state which version of the asset is used and, if possible, include a URL.
        \item The name of the license (e.g., CC-BY 4.0) should be included for each asset.
        \item For scraped data from a particular source (e.g., website), the copyright and terms of service of that source should be provided.
        \item If assets are released, the license, copyright information, and terms of use in the package should be provided. For popular datasets, \url{paperswithcode.com/datasets} has curated licenses for some datasets. Their licensing guide can help determine the license of a dataset.
        \item For existing datasets that are re-packaged, both the original license and the license of the derived asset (if it has changed) should be provided.
        \item If this information is not available online, the authors are encouraged to reach out to the asset's creators.
    \end{itemize}

\item {\bf New Assets}
    \item[] Question: Are new assets introduced in the paper well documented and is the documentation provided alongside the assets?
    \item[] Answer: \answerYes{} 
    \item[] Justification: We provide documentation and comments in the code.
    \item[] Guidelines:
    \begin{itemize}
        \item The answer NA means that the paper does not release new assets.
        \item Researchers should communicate the details of the dataset/code/model as part of their submissions via structured templates. This includes details about training, license, limitations, etc. 
        \item The paper should discuss whether and how consent was obtained from people whose asset is used.
        \item At submission time, remember to anonymize your assets (if applicable). You can either create an anonymized URL or include an anonymized zip file.
    \end{itemize}

\item {\bf Crowdsourcing and Research with Human Subjects}
    \item[] Question: For crowdsourcing experiments and research with human subjects, does the paper include the full text of instructions given to participants and screenshots, if applicable, as well as details about compensation (if any)? 
    \item[] Answer: \answerNA{} 
    \item[] Justification: Our paper did not involve this.
    \item[] Guidelines:
    \begin{itemize}
        \item The answer NA means that the paper does not involve crowdsourcing nor research with human subjects.
        \item Including this information in the supplemental material is fine, but if the main contribution of the paper involves human subjects, then as much detail as possible should be included in the main paper. 
        \item According to the NeurIPS Code of Ethics, workers involved in data collection, curation, or other labor should be paid at least the minimum wage in the country of the data collector. 
    \end{itemize}

\item {\bf Institutional Review Board (IRB) Approvals or Equivalent for Research with Human Subjects}
    \item[] Question: Does the paper describe potential risks incurred by study participants, whether such risks were disclosed to the subjects, and whether Institutional Review Board (IRB) approvals (or an equivalent approval/review based on the requirements of your country or institution) were obtained?
    \item[] Answer: \answerNA{} 
    \item[] Justification: Our paper did not invovle this.
    \item[] Guidelines:
    \begin{itemize}
        \item The answer NA means that the paper does not involve crowdsourcing nor research with human subjects.
        \item Depending on the country in which research is conducted, IRB approval (or equivalent) may be required for any human subjects research. If you obtained IRB approval, you should clearly state this in the paper. 
        \item We recognize that the procedures for this may vary significantly between institutions and locations, and we expect authors to adhere to the NeurIPS Code of Ethics and the guidelines for their institution. 
        \item For initial submissions, do not include any information that would break anonymity (if applicable), such as the institution conducting the review.
    \end{itemize}

\end{enumerate}

}{}

\end{document}